\documentclass[12pt,a4paper,reqno]{amsart}

\usepackage[utf8]{inputenc}
\usepackage[T1]{fontenc}
\usepackage[english]{babel}
\usepackage{amsmath}
\usepackage{amsfonts}
\usepackage{amssymb}
\usepackage{amsthm}
\usepackage{ucs}
\usepackage{graphicx}
\usepackage{xcolor}
\usepackage{dsfont}
\usepackage{tikz}
\usepackage{wasysym}
\usepackage{physics}
\usepackage{mathtools}
\usepackage{xifthen}
\usepackage{todonotes}
\usepackage{enumitem}
\usepackage{stmaryrd}
\usepackage{constants}
\usepackage[ruled]{algorithm2e}
\usepackage{tikz,pgfplots}
\usepackage{chngcntr}
\usepackage{hyperref}
\usepackage[font={small}]{caption}

\counterwithin{figure}{section}

\graphicspath{{Pics/}}

\usepackage{constants}

\newconstantfamily{ctrlcst}{
	symbol=\kappa
}
\newconstantfamily{csts}{
	symbol=C
}
\renewconstantfamily{normal}{
	symbol=c
}

\makeatletter
\newcommand{\pushright}[1]{\ifmeasuring@#1\else\omit\hfill$\displaystyle#1$\fi\ignorespaces}
\newcommand{\pushleft}[1]{\ifmeasuring@#1\else\omit$\displaystyle#1$\hfill\fi\ignorespaces}
\makeatother

\newcommand{\sphere}{\mathbb{S}}
\newcommand{\R}{\mathbb{R}}
\newcommand{\Rd}{\mathbb{R}^d}
\newcommand{\Zd}{\mathbb{Z}^d}

\newcommand{\Z}{\mathbb{Z}}
\newcommand{\N}{\mathbb{Z}_{+}}

\newcommand{\interior}{\mathrm{\scriptscriptstyle int}}
\newcommand{\exterior}{\mathrm{\scriptscriptstyle ext}}
\newcommand{\edge}{\mathrm{\scriptscriptstyle edge}}

\renewcommand{\norm}[1]{\|#1\|}

\newcommand{\given}{\,|\,}

\newcommand{\betac}{\beta_{\mathrm{\scriptscriptstyle c}}}

\newcommand{\icl}{\xi}

\newcommand{\cone}{\mathcal{Y}}
\newcommand{\fcone}{\mathcal{Y}^\blacktriangleleft}
\newcommand{\bcone}{\mathcal{Y}^\blacktriangleright}
\newcommand{\bend}{\mathbf{b}}
\newcommand{\fend}{\mathbf{f}}

\newcommand{\ConePts}{\mathrm{CPts}}

\newcommand{\displace}{D}

\newcommand{\WulffShape}{\mathbf{K}_{\xi}}
\newcommand{\equiDecSet}{\mathbf{U}_{\icl}}
\newcommand{\CGball}{\mathbf{B}}
\newcommand{\CGballK}{\mathbf{B}_K}
\newcommand{\CGballKFat}{\bar{\mathbf{B}}_K}

\newcommand{\treeMap}{\mathcal{T}}
\newcommand{\tree}{\mathfrak{T}}
\newcommand{\trunk}{\mathfrak{t}}
\newcommand{\branches}{\mathfrak{b}}
\newcommand{\treeWeight}{\Psi}

\newcommand{\good}{\mathrm{good}}

\newcommand{\procMem}{\Phi}

\newcommand{\SetRootMarkBackCont}{\mathfrak{B}_L}
\newcommand{\SetRootMarkBackIrr}{\mathfrak{B}_L^{\mathrm{irr}}}
\newcommand{\SetRootMarkForwCont}{\mathfrak{B}_R}
\newcommand{\SetRootMarkForwIrr}{\mathfrak{B}_R^{\mathrm{irr}}}
\newcommand{\SetRootDiaCont}{\mathfrak{A}}
\newcommand{\SetRootDiaIrr}{\mathfrak{A}^{\mathrm{irr}}}

\newcommand{\concatenate}{\circ}

\newcommand{\bfn}{\mathbf{n}}

\newcommand{\bfm}{\mathbf{m}}
\newcommand{\weight}{\textnormal{w}}

\newcommand{\RCPF}{Z}
\newcommand{\evenPart}{\mathcal{E}}
\newcommand{\RCLaw}{\mathbf{P}}

\newcommand{\free}{\mathrm{\scriptscriptstyle f}}

\renewcommand{\nleftrightarrow}{\mathrel{\ooalign{$\leftrightarrow$\cr\hidewidth$/$\hidewidth}}}
\newcommand{\nlongleftrightarrow}{\mathrel{\ooalign{$\longleftrightarrow$\cr\hidewidth$/$\hidewidth}}}

\newcommand{\bfp}{\mathbf{p}}
\newcommand{\bfe}{\mathbf{e}}
\newcommand{\bfa}{\mathbf{a}}
\newcommand{\bfd}{\mathbf{d}}
\newcommand{\bfJ}{\mathbf{J}}

\newcommand{\bfs}{\mathbf{s}}
\newcommand{\calA}{\mathcal{A}}

\newcommand{\calC}{\mathcal{C}}

\theoremstyle{plain}
\newtheorem{theorem}{Theorem}[section]
\newtheorem{lemma}[theorem]{Lemma}

\theoremstyle{remark}

\theoremstyle{definition}

\newtheoremstyle{case}{}{}{}{}{}{:}{ }{}
\theoremstyle{case}
\newtheorem{case}{Case}

\author{S\'{e}bastien Ott}

\title[Sharp Asymptotics of \(\langle\sigma_0;\sigma_x\rangle\) when \(h>0\).]{Sharp Asymptotics for the Truncated Two-Point Function of the Ising Model with a Positive Field}

\begin{document}
	
\begin{abstract}
We prove that the correction to exponential decay of the truncated two points function in the homogeneous positive field Ising model is \(c\norm{x}^{-(d-1)/2}\). The proof is based on the development in the random current representation of a ``modern'' Ornstein-Zernike theory, as developed by Campanino, Ioffe and Velenik \cite{Campanino+Ioffe+Velenik-2008}.
\end{abstract}

\maketitle

\section{Introduction}

\subsection{Ising Model with a Positive Field}
For all this paper, \(h\) will be a fixed positive real number and will thereof be omitted from notations. We will also absorb the inverse temperature \(\beta\) in the coupling constants \(\bfJ\).

The positive field Ising Model on a weighted graph \(G=(V_G,\bfJ)\) is defined via:
\[
\mu_G(\sigma) \propto \exp{\sum_{\{i,j\}\subset V_G} J_{ij}\sigma_i\sigma_j + h\sum_{i}\sigma_i}
\]
with \(h>0\). Denote the expectation \(\langle\cdot \rangle_G\).

As \(h\geq 0\) and \(J_{ij}\geq 0\), the measure \(\mu_G\) satisfies the GKS inequality:

\begin{lemma}
	For any \(A,B\subset V_G\),
	\begin{gather*}
	\langle\sigma_A\rangle_G\geq 0,\\
	\langle\sigma_A\sigma_B\rangle_G\geq \langle\sigma_A\rangle_G\langle\sigma_B\rangle_G.
	\end{gather*}
\end{lemma}

In what follows, we will be considering the Ising model with \(h>0\) on \(\Zd\): the vertex set is the sites of \(\Z^d\) (canonically embedded in \(\R^d\)) and the weights \((J_{ij})_{\{i,j\}\subset\Z^d}\) are assumed to be invariant under translation by elements of \(\Z^d\) and under symmetries of \(\Zd\). Moreover, we suppose that there exists \(R>0\) such that \(J_{ij}=0,\forall \norm{i-j}>R\) (finite range). Denote \(d(\cdot ,\cdot)\) the graph distance in \((\Zd,\bfJ)\).

For technical convenience, we will assume \(J_{0i}>0\) for \(i\in\{\pm \bfe_k, k=1,... ,d\}\), where \(\bfe_k\) the unit vector in the \(k\)th direction.

Let \(\mu\) be the (unique as \(h>0\), for example using Lee-Yang Theorem, see \cite{Friedli+Velenik-2017}) infinite volume measure obtained by limit of finite volume measures \(\mu_G\). Expectation under \(\mu\) is denoted \(\langle\ \rangle \).

The object of study in this paper will be the truncated two point function:
\[
\langle\sigma_x;\sigma_y\rangle = \langle\sigma_x\sigma_y\rangle-\langle\sigma_x\rangle\langle\sigma_y\rangle.
\]

We first define the inverse correlation length via:
\begin{equation}
\label{eq:finVolInvCorLen}
\icl_n(x) = -\frac{1}{n}\log( \langle\sigma_0; \sigma_{[nx]}\rangle)
\end{equation}
(with \([x]\) denote the point in \(\Zd\) closest to \(x\)) and
\begin{equation}
\label{eq:InvCorLen}
\icl(x) = \lim_{n\to \infty}\icl_n(x)
\end{equation}
if it exists.

The following Theorem (which is a combination of result from \cite{graham_correlation_1982} and well known arguments) encapsulates the basic properties of \(\icl\) that will be needed for our study.

\begin{theorem}
	\label{thm:iclProperties}
	\(\icl(x) \) exists for every \(x\in \R^d\) and
	\begin{equation}
	\langle \sigma_0;\sigma_x \rangle\leq e^{-\icl(x)}.
	\end{equation}
	Moreover,
	\begin{equation}
	\langle \sigma_0;\sigma_x \rangle\leq \langle \sigma_0\sigma_x \rangle\cosh(h)^{-d(0,x)},
	\end{equation}
	and \(\icl\) defines a norm on \(\R^d\).
\end{theorem}

The main goal of this paper is the proof of:
\begin{theorem}
	\label{thm:OZprefactor}
	For any \(d\geq 1\), \(\beta>0\) and \(h>0\), there exists \(\psi_d:\sphere^{d-1}\to\R_{\geq 0}\) such that
	\begin{equation*}
	\langle \sigma_0;\sigma_x \rangle = \frac{\psi_d(x/\norm{x})}{\norm{x}^{(d-1)/2}} e^{-\icl(x)}(1+o_{\norm{x}}(1)).
	\end{equation*}
	Moreover, both \(\icl\) and \(\psi_d\) are analytic in \(x/\norm{x}\).
\end{theorem}

The claim being trivial when \(d=1\). Theorem~\ref{thm:OZprefactor} will be proved as a corollary of Theorem~\ref{thm:mainThm}.

\subsection{Overview of the Proof}

The main result of this paper is the development of an Ornstein-Zernike theory, in the form introduced in \cite{Campanino+Ioffe-2002},\cite{Campanino+Ioffe+Velenik-2003} and \cite{Campanino+Ioffe+Velenik-2008}, in the double random current representation of the Ising model. Such theory is already available in the random-path representation of correlations induced by the high-temperature expansion of the Ising model and in the Random Cluster Model. The OZ construction is done in three steps: the first is a coarse graining argument that allows to approximate long connected objects by a family of random trees embedded in \(\Zd\) satisfying certain exponential cost inequality (Sections~\ref{sec:CoarseGraining} and~\ref{sec:TreeEnergyExtrac}). Then, a very robust procedure developed in \cite{Campanino+Ioffe-2002},\cite{Campanino+Ioffe+Velenik-2003} and \cite{Campanino+Ioffe+Velenik-2008} allows to control the geometry of those trees; claims following from this procedure will be entirely imported from \cite{Campanino+Ioffe+Velenik-2008}, the present work is thus not self-contained. Finally, a finite-energy type argument (Section~\ref{sec:FEandFactorization}) transfers the control obtained on the trees to a control on the initial object, giving a representation of the cluster in terms of a concatenation of smaller ``irreducible'' clusters; the proof finishes by the construction (imported from \cite{Ott2018}) of a probability measure on a set of small clusters that are concatenations of irreducible ones so that the concatenation of an i.i.d. sequence of clusters has the same law as the long connected object one started with (see Theorem~\ref{thm:mainThm}).

\subsection{Brief History of Sharp Asymptotics}

Asymptotic study of covariances goes back to the (non-rigorous) work of Ornstein and Zernike \cite{Ornstein+Zernike-1914},\cite{Zernike-1916}. First rigorous analysis for truncated two-point functions in any dimension were given in \cite{Abraham+Kunz-1977} and \cite{Paes-Leme-1978} in the regime \(h=0, \beta\ll 1\). The first non-perturbative treatment of the question was done in \cite{Campanino+Ioffe+Velenik-2003}, in the regime \(h=0, \beta<\betac\); it was then extended to Potts models (for the same regime) in \cite{Campanino+Ioffe+Velenik-2008}. In all non-perturbative treatments, the analysis is possible because the truncation is ``trivial'': \(\langle\sigma_0\rangle=0\) and the study is thus reduced to the study of the two-point function \(\langle\sigma_0\sigma_x\rangle\) which enjoys nicer graphical representation than \(\langle\sigma_0;\sigma_x\rangle\). The only sharp, non-perturbative, study of \(\langle\sigma_0;\sigma_x\rangle\) with a non-trivial truncation is the nearest-neighbour Ising model on \(\Z^2\), with \(h=0\) and any \(\beta\), as it is possible to derive an exact formula for it using integrability (see for example \cite{mccoy_two-dimensional_1973}).

The first non-pertubative study of covariances with a non-trivial truncation not relying on integrability was made in \cite{Ott+Velenik-2018(2)} for covariances between even products of spins, in the regime \(\beta<\betac\). It is largely based on the random current representation of the Ising model. The present work uses the same representation to attack the question in the regime \(\beta>0, h>0\) (and, by symmetry, \(h<0\)).

\subsection{Open Problems}
We list here a few open problems that appear to be natural given the results of the present work and those available from past ones.

\subsubsection*{Low Temperature Covariances} The first problem suggested by the results of this paper, of \cite{Campanino+Ioffe+Velenik-2003} and of \cite{mccoy_two-dimensional_1973}, is the sharp treatment of truncated correlations when \(h=0,\beta>\betac, d\geq 3\), as it is the only regime missing. Such an analysis seems doable via the arguments presented in \cite{duminil-copin_exponential_2018} combined with the general approach of \cite{Campanino+Ioffe+Velenik-2008} and some ideas borrowed from the present work. We plan to come back to this question in a near future.

In the case of nearest neighbour interaction, treating the case \(d=2\) without resorting to integrability should be doable using planar duality and the analysis done in \cite{Ott+Velenik-2018(2)}. Extension to general finite range interaction is wide open.

\subsubsection*{Extension to Potts Model} A second problem (of apparently higher level of difficulty) is the treatment of truncated two-point functions in the Potts model with an homogeneous field, the absence of random current representation asks for new ideas to even start the analysis. A good indication that the question is highly non-trivial are the results of \cite{biskup_coexistence_2000}.

\subsubsection*{Quantum Ising Model with Transverse Field} A third problem of interest would be a treatment of the asymptotic behaviour of truncated two-point functions for the quantum Ising model with transverse magnetic field (a random current representation satisfying a switching lemma being available -see \cite{Crawford2010}- this seems to be the easiest problem of the list).

\section{Random Current Representation, Notation and Main Theorem}

\subsection{Notation and Conventions}

Start with a few generic notations. Denote \(o_n(1)\) a quantity that tends to \(0\) as \(n\) goes to infinity. Constants \(c,C\) are non-negative real numbers that can vary from line to line. They do not depend on the variables of interest in the study. We also fix an arbitrary total order on \(\Zd\) for the whole paper. We write \(\N=\{0,1,2,3,... \}\).

We work on weighted graphs; a weighted graph \(G\) will be the data of a vertex set \(V_G\) and a weight function \(\bfJ:\big\{\{i,j\}\subset V_G\big\}\to\R_{\geq 0}\). Write \(\bfJ(\{i,j\})=J_{ij}\). The edges of \(G\), denoted \(E_G\), are the pairs \(\{i,j\}\) with non-zero weight: \(E_G=\{ \{i,j\}\subset V_G:\ J_{ij}>0 \}\). We write \(i\sim j\) if \(J_{ij}>0\).

We say that two points \(u,v\in V_G\) are \emph{connected}, written \(u\leftrightarrow v\), if there exists a path of edges in \(E_G\) going from \(u\) to \(v\). Call the \emph{cluster} of \(v\in V_G\) the maximal connected component of \(v\).

For \(G\) a graph and \(V_H\subset V_G\), define \(H\) the sub-graph of \(G\) induced by \(V_H\) to be \(H=(V_H,\bfJ|_{\{i,j\}\subset V_H})\). Write \(H\subset G\). For \( A\subset V_G\), define three notions of boundary for \(A\):
\begin{gather*}
\partial^{\exterior} A =\big\{v\in V_G\setminus A: \exists v'\in A, \{v,v'\}\in E_G \big\},\\
\partial^{\interior} A =\big\{v\in A: \exists v'\in V_G\setminus A, \{v,v'\}\in E_G \big\},\\
\partial^{\edge} A =\big\{ \{v,v'\}\in E_G: v\in A, v'\in V_G\setminus A \big\}.
\end{gather*}

For \(A\subset\R^d\), \(\partial A\) denotes the usual boundary of \(A\). A \emph{contour} \(\Gamma\subset E_G\) is a set of edges such that there exists \(V_H\subset V_G\) with \(H\) connected such that \(\partial^{\edge} V_H =\Gamma\). We will consider graphs to be marked graphs (having a distinguished vertex). When \(G\) is a subset of \(\Zd\) containing \(0\), the marked vertex will be \(0\) by convention. When \(V_H\) contains the marked vertex, it is called the \emph{interior} of \(\Gamma\), written \( \mathring{\Gamma} \), and \(V_G\setminus V_H\) the \emph{exterior}. Finally, for a contour \(\Gamma\subset E_G\), define
\begin{gather*}
	\Gamma^{\interior} = \big\{v\in \mathring{\Gamma}: \exists v'\in V_G\setminus\mathring{\Gamma}, \{v,v'\}\in \Gamma \big\},\\
	\Gamma^{\exterior} = \big\{v\in V_G\setminus\mathring{\Gamma}: \exists v'\in \mathring{\Gamma}, \{v,v'\}\in \Gamma \big\}.
\end{gather*}

The last convention is that we assimilate subsets \(A\subset B\) with their characteristic function \(\mathds{1}_A: B\to\{0,1\}\) (in particular, percolation configurations \(\omega\subset E_G\) are seen as functions via \(\omega_e=\mathds{1}_{e\in \omega} \)).

\subsection{Random Current Model}

Let \(G=(V_G,\bfJ)\) be a weighted graph. A \emph{random current configuration} on \(G\) is an element of \(\N^{E_G}\). The \emph{weight} of a configuration \(\bfn\) is given by:
\begin{equation*}
	\weight(\bfn) = \prod_{e\in E_G} \frac{(J_e)^{\bfn_e}}{\bfn_e!}.
\end{equation*}
The \emph{sources} of a current are:
\begin{equation*}
\partial\bfn = \{v\in V_{G}: \sum_{u\sim v} \bfn_{uv} = 1\mod 2\}.
\end{equation*}
Notice that the number of sources of a current is always even. We will use the shorthand \(\sum_{\partial \bfn = A}\) for the sum over current configurations with sources \(A\). Define then the associated partition functions and probability measures:
\begin{equation*}
	\RCPF_{G}(A) = \sum_{\partial \bfn = A} \weight(\bfn),\qquad \RCLaw_{G}^A(\bfn) = \mathds{1}_{\partial\bfn = A}\frac{\weight(\bfn )}{\RCPF_{G}(A)}.
\end{equation*}
We will use the notation \(\RCLaw_{G}^{A,B}=\RCLaw_{G}^A\times \RCLaw_{G}^B\), called the \emph{double random current} measure. We dealing with pairs of random current, the following notation will be useful: for \(F,f,g:\N^{E_G}\to \R\), denote
\begin{equation*}
	\RCPF_{G}(A)\RCPF_{G}(B)\big\{ F f(\bfn) g(\bfm)\big\} = \sum_{\partial\bfn=A} \sum_{\partial\bfm=B} \weight(\bfn)\weight(\bfm)F(\bfn+\bfm)f(\bfn) g(\bfm),
\end{equation*}where the sum is coordinatewise.

\subsection{Random Current Representation of the Ising Model with a Field}

Start by defining the augmented graph \(G_g=(V_{G_g},\tilde{\bfJ})\) by adding a vertex \(g\) to \(V_G\) and setting \(\tilde{J}_{ig}=h\) for all \(i\in V_G\) and \(\tilde{J}_{ij}=J_{ij}\) for all \(i,j\in V_G\). Write \(E_{G_g}=\big\{ \{i,j\}\subset V_{G_g}: \tilde{J}_{ij}>0 \big\}\). One can thus rewrite:
\[
\mu_{G}(\sigma) \propto \mathds{1}_{\sigma_g=1}\exp\Big( \sum_{\{i,j\}\subset V_{G_g}}\tilde{J}_{ij}\sigma_i\sigma_j \Big).
\]

The random current representation is obtained by expanding \(e^{ \tilde{J}_{ij}\sigma_i\sigma_j }\) and resumming; more precisely, for any \(A\subset V_G\), one gets

\begin{multline*}
\sum_{\sigma\in\{\pm 1\}^{V_{G_g}} } \sigma_A\mathds{1}_{\sigma_g=1} e^{ \sum_{\{i,j\}\in E_{G_g} }\tilde{J}_{ij}\sigma_i\sigma_j } = \\
= 2^{|V_G|}\begin{cases}
\sum_{\partial\bfn=A} \weight(\bfn) & \text{ if }|A|=0\mod{2}\\
\sum_{\partial\bfn =A\cup\{g\} } \weight(\bfn) & \text{ if }|A|=1\mod{2}
\end{cases}.
\end{multline*}
Thus,
\begin{equation*}
\langle \sigma_A \rangle_{G} = \frac{\RCPF_{G_g}(A)}{\RCPF_{G_g}(\varnothing)} \text{ if } |A|\text{ even},\quad  = \frac{\RCPF_{G_g}(A\cup \{g\})}{\RCPF_{G_g}(\varnothing)} \text{ if } |A|\text{ odd}.
\end{equation*}

One can define infinite volume random current measures on \(G_g\) for \(G=(\Zd,\bfJ)\), \(\RCLaw_{G_g}^A\), by taking weak limits, see \cite{aizenman_random_2015}. For fixed sources, the limit will be unique by Lemma~\ref{thm:expMixRC}.

\subsection{Connectivity in Random Current and Switching Lemma}

Now we present two classical marginals of the random current and of the double random current. The first one is the percolative interpretation of the current via \(\hat{\bfn}_e=\mathds{1}_{\bfn_e>0}\) for all edges \(e\).

The main use of this point of view is the Switching Lemma. For \(A\subset V_{G_g}\) define \(\evenPart_A\) the event that every cluster contains an even number of sites in \(A\) (possibly \(0\)).

\begin{lemma}
	\label{lem:RCswitchingLemma} For any graph \(G=(V_G,\bfJ)\) and \(A,B\subset V_G \),
	\[\RCPF_{G}(A)\RCPF_{G}(B)\big\{ F\big\} = \RCPF_{G}(A\Delta B)\RCPF_{G}(\varnothing)\big\{ \mathds{1}_{\evenPart_B}F\big\}\]
\end{lemma}
It first appeared in \cite{griffiths_concavity_1970} and was used extensively in \cite{aizenman_geometric_1982}, a more recent demonstration of the random current efficiency is \cite{aizenman_random_2015}.

The second useful marginal is to forget about the number associated to each edge and to remember only whether its current number is even \(>0\), \(=0\) or odd. Thus define
\begin{equation*}
	\bar{\bfn}_e= \begin{cases}
	0 & \text{ if } \bfn_e=0\\
	1 & \text{ if } \bfn_e=1\mod 2\\
	2 & \text{ else}
	\end{cases}.
\end{equation*}
Notice that \(\partial\bfn=\partial\bar{\bfn}\). We call such functions \emph{parity currents}.

The weight of \(\bar{\bfn}\) is:
\begin{equation*}
	\weight(\bar{\bfn}) = \sum_{\bfn\sim\bar{\bfn}} \weight(\bfn)
	= \prod_{e:\bar{\bfn}_e=1} \sinh(J_e)\prod_{e:\bar{\bfn}_e=2}\cosh(J_e)
\end{equation*}
where \(\bfn\sim\bar{\bfn}\) means the compatibility of \(\bfn\) and \(\bar{\bfn}\). The first application of the switching lemma we will use is a percolative representation of the truncated two-point function:
\begin{multline*}
\label{eq:truncated_percoRep}
	\langle \sigma_0;\sigma_x \rangle_{G} = \frac{\RCPF_{G_g}(0,x) \RCPF_{G_g}(\varnothing) - \RCPF_{G_g}(0,g)\RCPF_{G_g}(x,g)}{\RCPF_{G_g}(\varnothing)^2}=\\
	=\frac{\RCPF_{G_g}(\varnothing) \RCPF_{G_g}(0,x)\{\mathds{1}_{0\nleftrightarrow g} \} }{\RCPF_{G_g}(\varnothing)^2} = \langle \sigma_0\sigma_x \rangle_{G} \RCLaw_{G_g}^{\varnothing,\{0,x\}}(0\nleftrightarrow g).
\end{multline*}
The reader not used to working with random current can notice that a connexion between \(0\) and \(x\) is induced in the second current via the source constraint.

\subsection{Diamonds Decomposition and Main Theorem}
\label{subsec:DiamondDecomp_and_MainThm}

The study of \(\langle \sigma_0;\sigma_x \rangle\) will be done by studying the behaviour of a long cluster in a well chosen percolation model. It will follows the study done in \cite{Campanino+Ioffe+Velenik-2008} for FK-percolation and the main goal of this paper is to obtain a similar decomposition of the cluster in terms of ``cone-confined'' components. Before stating the main result of this paper, we need to introduce a bit of notation and a few objects. We will need two convex sets encoding the information about \(\icl\): the \emph{equi-decay set} \(\equiDecSet\) and the \emph{Wulff shape} \(\WulffShape\) defined by
\begin{gather*}
\equiDecSet = \{x\in\R^d: \icl(x)\leq 1 \},\qquad \WulffShape = \bigcap_{\bfs\in\sphere^{d-1}} \{t\in \R^d: (t,\bfs)_d\leq \icl(\bfs) \},
\end{gather*}
where \((\cdot,\cdot)_d \) denotes the scalar product on \(\R^d\). They are polar:
\begin{gather*}
\equiDecSet = \big\{x\in\R^d: \max_{t\in\WulffShape} (t,x)_d\leq 1 \big\},\qquad \WulffShape = \big\{t\in \R^d: \max_{x\in\equiDecSet}(t,x)_d\leq 1 \big\}.
\end{gather*}
\(x\in\R^d\) and \(t\in\partial\WulffShape\) are said to be dual if \((t,x)_d=\icl(x)\). For \(t\in\partial\WulffShape\) and \(\delta\in(0,1)\), define the forward and backward cones
\begin{equation*}
\fcone_{\delta}(t)=\big\{ x\in\Zd: (t,x)_d>(1-\delta)\icl(x) \big\},\qquad \bcone_{\delta}(t) = -\fcone_{\delta}(t).
\end{equation*}
For a set \(A\subset \Rd\), we say that \(v\in A\) is a \((t,\delta)\)-\emph{cone-point} of \(A\) if
\begin{equation*}
	A\setminus v \subset \big( (v+\fcone_{\delta}(t))\cup (v+\bcone_{\delta}(t)) \big).
\end{equation*}
For a connected set \(A\subset \Rd\), we say that \(v\in A\) is a \emph{break point} of \(A\) if \(A\setminus v=A_1\cup A_2\) with \(A_1\cap A_2=\varnothing\) and \(A_1,A_2\) connected. For \(v\) a break point of \(A\), we say that \(v\) is a
\begin{itemize}
	\item \((t,\delta)\)-\emph{forward cone-point} of \(A\) if \(A_1\subset (v+\fcone_{\delta}(t))\) and \(A_2\subset (v+\overline{\fcone_{\delta}(t)} )^c\) (we take the closure of \(\fcone_{\delta}(t)\) as it is an open set) or the same with \(A_1,A_2\) interchanged,
	\item \((t,\delta)\)-\emph{backward cone-point} of \(A\) if \(A_1\subset (v+\bcone_{\delta}(t))\) and \(A_2\subset (v+\overline{\bcone_{\delta}(t)} )^c\) or the same with \(A_1,A_2\) interchanged.
\end{itemize}
A connected set \( A \subset \Rd\) is called
\begin{itemize}
	\item \((t,\delta)\)-\emph{forward} (resp. \emph{backward})-\emph{contained} if there exists \(v\in A\) such that \(A\setminus v \subset(v+\fcone_{\delta}(t))\), resp. \(\subset (v+\bcone_{\delta}(t))\). Denote this \(v= \fend(A)\) (resp. \( =\bend(A)\)).
	\item \((t,\delta)\)-\emph{diamond-contained} if it is forward and backward contained.
	\item \((t,\delta)\)-\emph{forward-irreducible} if there exists \(v\in A\) such that \(A \subset (v+\fcone_{\delta}(t))\) and \(v\) is the only forward cone-point of \(A\).
	\item \((t,\delta)\)-\emph{backward-irreducible} if there exists \(v\in A\) such that \(A\subset (v+\bcone_{\delta}(t))\) and \(v\) is the only backward cone-point of \(A\).
	\item \((t,\delta)\)-\emph{irreducible} if it is forward and backward irreducible.
\end{itemize}
These definitions are adapted to \(C\) a cluster in \((\Zd,\bfJ)\) and \(v\) a vertex in \(C\) via the use of the above definitions on the subset of \(\Rd\) obtained by taking the vertices of \(C\) and the solid line segments corresponding to edges.

Notice that any diamond contained cluster is a concatenation of irreducible ones. Introduce then the notion of \emph{displacement} of a cluster by: for any \(C\) diamond-contained,
\begin{equation}
\displace(C) = \bend(C)-\fend(C).
\end{equation}
To make sense of the displacement for not diamond-contained cluster, consider the set of \emph{marked} forward/backward contained cluster (the set of forw./backw.-contained cluster with a distinguished vertex \(v_*\)). In this case, define the displacement of \((C,v_*)\) by:
\begin{equation}
v_*-\fend(C), \qquad \text{ resp. }\qquad \bend(C)-v_*.
\end{equation}
We will also say that a marked forward/backward-contained cluster \(\gamma\) is irreducible if it contains no cone-point in \(v_* + \bcone_{\delta}(t)\), resp. \(v_* + \fcone_{\delta}(t)\).

Define
\begin{itemize}
	\item The sets of rooted marked backward-contained/irreducible clusters:
	\begin{gather*}
	\SetRootMarkBackCont = \{C \text{ marked backward-contained with } v^* =0 \},\\
	\SetRootMarkBackIrr = \{C \text{ marked backward-irreducible with } v^* =0 \}.
	\end{gather*}
	\item The sets of rooted marked forward-contained/irreducible clusters:
	\begin{gather*}
	\SetRootMarkForwCont = \{C \text{ marked forward-contained with } \fend(C) =0 \},\\
	\SetRootMarkForwIrr = \{C \text{ marked forward-irreducible with } \fend(C) =0 \}.
	\end{gather*}
	\item The sets of rooted diamond-contained/irreducible clusters:
	\begin{gather*}
	\SetRootDiaCont = \{C \text{ diamond-contained with } \fend(C) =0 \},\\
	\SetRootDiaIrr = \{C \text{ irreducible with } \fend(C) =0 \}.
	\end{gather*}
\end{itemize}

For \(\gamma\in \SetRootMarkBackCont \) and \(\gamma'\in \SetRootMarkForwCont\) define the \emph{concatenation} of \(\gamma\) and \(\gamma'\) by:
\begin{equation*}
	\gamma \concatenate \gamma' = \gamma \cup (\bend(\gamma) + \gamma').
\end{equation*}

We will often identify a chain of irreducible clusters \(\gamma_L\sqcup\gamma_1\sqcup\cdots \sqcup \gamma_R\) with the concatenation of their translates. The displacement of a concatenation of clusters is the sum of the displacements.

\begin{figure}[h]
	\centering
	\includegraphics[height=75pt, width=300pt]{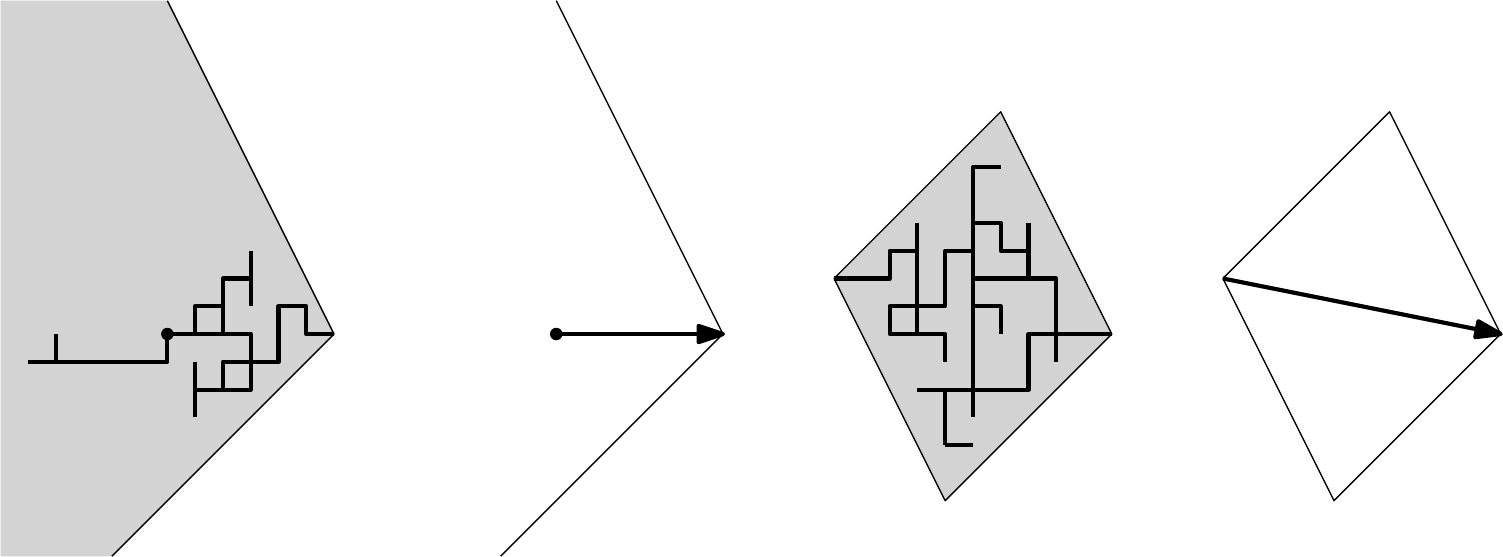}
	\caption{From left to right: a marked irreducible piece, its displacement, an irreducible piece, its displacement.}
	\label{fig:Irred_Displace}
\end{figure}

The main result of this paper is
\begin{theorem}
	\label{thm:mainThm}
	For any \(t_0\in\partial\WulffShape\) there exists \(\delta\in(0,1)\) such that for any \(t\in\partial\WulffShape\) in an small open neighbourhood of \(t_0\) and any \(x\) dual to \(t\), one can construct two non-negative, finite measures \(\rho_L,\rho_R\) on \(\SetRootMarkBackCont\) and \(\SetRootMarkForwCont\) and a probability measure \(\bfp\) on \(\SetRootDiaCont\) (for the cones \(\cone_{\delta}(t_0)\)) satisfying: there exists \(c>0\) such that, for any \(f\) measurable with respect to the cluster of \(0\),
	\begin{multline*}
	\Big|h^2\langle\sigma_0\sigma_x\rangle e^{(t_0,x)}\sum_{C} \RCLaw^{\varnothing,\{0,x\}}\big( C_{0,x}=C, 0\nleftrightarrow g \big)f(C) - \\
	-\sum_{\gamma_L}\sum_{\gamma_R}\rho_L(\gamma_L)\rho_{R}(\gamma_R)\sum_{M\geq 0}\sum_{\gamma_1,... ,\gamma_M}\prod_{i=1}^{M} \bfp(\gamma_i) \mathds{1}_{D(\gamma)=x}  f(\gamma)\Big|\leq \norm{f}_{\infty} e^{-c\norm{x}}
	\end{multline*}
	where \(\gamma = \gamma_L\concatenate\gamma_1\concatenate\cdots \concatenate\gamma_M\concatenate\gamma_R\). Moreover, there exists \(\nu>0\) such that
	\begin{equation}
	\label{eq:espDecFacMeas}
	\rho_L(\norm{D(\gamma_L)}\geq l) \vee \rho_R(\norm{D(\gamma_R)}\geq l) \vee \bfp(\norm{D(\gamma_1)}\geq l) \leq e^{-\nu l}.
	\end{equation}
\end{theorem}

From Theorem~\ref{thm:mainThm}, one easily deduces Theorem~\ref{thm:OZprefactor} by setting \(f(C)=1\) and using the local limit theorem for random walks in dimension \(d\) together with the exponential decay property of \(\rho_L,\rho_R\).

\section{Proof of Theorem~\ref{thm:iclProperties}}

Theorem~\ref{thm:iclProperties} follows from two classical lemmas.
\begin{lemma}
	For any graph \(G\) finite, \(h\geq 0\) and \(x,y,z\in V_G\),
	\begin{equation}
	\label{eq:subAddT2ptsFct}
	\langle\sigma_x; \sigma_z\rangle_G \geq \langle\sigma_x; \sigma_{y}\rangle_G \langle\sigma_y; \sigma_{z}\rangle_G.
	\end{equation}
	In particular, \(\icl(x)\) exists for every \(x\in\R^d\) and \(\langle \sigma_0;\sigma_x \rangle\leq e^{-\icl(x)}\).
\end{lemma}
\begin{proof}
	A proof of \eqref{eq:subAddT2ptsFct} is given in \cite{graham_correlation_1982}.
	The existence of \(\icl(x)\) and \(\langle \sigma_0;\sigma_x \rangle\leq e^{-\icl(x)}\) then follow from \eqref{eq:subAddT2ptsFct} and Fekete's lemma.
\end{proof}

\begin{lemma}
	\label{lem:PosICL}
	For any graph \(G\),
	\begin{equation*}
		\langle \sigma_0;\sigma_x \rangle_{G}\leq \langle \sigma_0\sigma_x \rangle_{G}\cosh(h)^{-d_{G}(0,x)}.
	\end{equation*}
\end{lemma}
\begin{proof}
	By the Switching Lemma,
	\begin{equation}
	\label{eq:percoRepTTPF}
	\frac{\langle \sigma_0;\sigma_x \rangle_{G}}{\langle \sigma_0\sigma_x \rangle_{G}} = \frac{\RCPF_{G_g}(\varnothing)\RCPF_{G_g}(0,x)\big\{ \mathds{1}_{0\nleftrightarrow g} \big\}}{\RCPF_{G_g}(\varnothing)\RCPF_{G_g}(0,x)}
	= \RCLaw_{G_g}^{\varnothing,\{0,x\}}(0\stackrel{\bfn+\bfm}{\nlongleftrightarrow g}).
	\end{equation}
	Now, notice that the source constraint in the first current imposes that \(0\stackrel{\bfm}{\longleftrightarrow} x\). Partitioning then with respect to the cluster of \(0\) in the second copy,
	\begin{align*}
	\RCLaw_{G_g}^{\varnothing,\{0,x\}}\big( 0\stackrel{\bfn+\bfm}{\nlongleftrightarrow} g \big) &\leq \sum_{\substack{C\subset V_G\\0,x\in C}} \RCLaw_{G_g}^{\{0,x\}}\big( C_{0,x}=C \big)\RCLaw_{G_g}^{\varnothing}\big( C\nleftrightarrow g \big)\\
	&\leq \cosh(h)^{-d_{G}(0,x)},
	\end{align*}by Lemma~\ref{lem:RCInsTol} and the fact that if \(C\) is the joint cluster of \(0\) and \(x\) in \(G\), it contains at least \(d_G(0,x)\) vertices.
\end{proof}

\(\icl\) being a norm follows from~\eqref{eq:subAddT2ptsFct} and Lemma~\ref{lem:PosICL}.

\section{Proof of Theorem~\ref{thm:OZprefactor}}

We deduce here Theorem~\ref{thm:OZprefactor} from Theorem~\ref{thm:mainThm}. We start by sketching the proof the OZ prefactor. Fix \(t\in\partial\WulffShape\) and \(x\) dual to \(t\). As said earlier, this is a rather straightforward application of the Local Limit Theorem in dimension \(d\). Denote \(\vec{\mu}= \mu\frac{x}{\norm{x}}\) the expectation of \(\displace(\gamma)\) under \(\bfp\) (where \(\gamma\sim \bfp\)). Fix \(\epsilon>0\) small. Denote \(P\) the law of the random walk with steps having law \(\displace_*(\bfp)\) (the push-forward of \(\bfp\) by \(\displace\)) started at \(0\). By the LLT, for any point \(y\) with \(\norm{y-n\vec{\mu}}<n^{1/2-\epsilon}\), \(P(\exists N: S_N = y) = \Cl{OZ1} n^{-(d-1)/2}(1+o_n(1))\) with \(\Cr{OZ1}\) depending only on \(\bfp\) (the way for \(\exists N: S_N... \) to \(n\) dependency is a classical concentration estimate). One can check the proof of the same statement in~\cite[Section 3.9]{Ott-2019} for the detailed argument. The claim follows from this and the exponential decay property of \(\rho_{L/R}\). One has:
\begin{equation*}
	\psi_d(x/\norm{x}) = \Cr{OZ1} \mu^{(d-1)/2}\sum_{u,v\in \fcone_{\delta}} \rho_L(D(\gamma_L)=u)\rho_R(D(\gamma_R)=v).
\end{equation*}

Now turn to the local analyticity of \(\icl\) in the direction (i.e.: analyticity as a function \(\sphere^{d-1}\to \R\)). By duality, it is equivalent to show that the boundary of the Wulff shape is locally analytic. \(\partial\WulffShape\) can be defined (see \cite{Campanino+Ioffe+Velenik-2003}) as the boundary of the convergence domain of
\begin{equation*}
	t \mapsto \sum_{x\in \Zd} e^{(t,x)_d}\RCLaw^{\varnothing,\{0,x\}}(0\nleftrightarrow g).
\end{equation*}
Fix \(t_0\in\partial\WulffShape\), \(s_0\in\sphere^{d-1}\) dual to \(t_0\) and let \(\delta,\bfp\) be given by Theorem~\ref{thm:mainThm}. Let \(V\subset \partial\WulffShape\) be a small open neighbourhood of \(t_0\) (that is chosen small enough so that all arguments go through). Let \(U\subset\sphere^{d-1}\) be the small open neighbourhood of \(s_0\) defined as the set of points dual to a point in \(V\). By definition of \(\icl\), \(t\in V\) is equivalent to \(t\) being in the boundary of the convergence domain of
\begin{equation*}
	A(z) = \sum_{x\in\cone(U)} e^{(x,z)_d} \RCLaw^{\varnothing,\{0,x\}}(0\nleftrightarrow g),
\end{equation*}where \(\cone(U)\) is the convex cone generated by \(U\). Choosing \(V\) (and thus \(U\)) sufficiently small so that Theorem~\ref{thm:mainThm} applies, one get that the boundary of the convergence domain of \(A(z)\) close to \(t_0\) is given by the set of points \(t_0+z\) with \(z\) in the boundary of the convergence domain of
\begin{equation*}
	\tilde{F}(z) = \sum_{x\in\cone(U)} e^{(x,z)_d} P_0(\exists N:\ S_N=x)
\end{equation*}
where \(P\) is the same random walk law as defined earlier in this Section. As a sum of i.i.d. random variables, \(S\) admits a large deviation principle. The convergence domain of \(\tilde{F}\) is thus the same as the convergence domain of
\begin{equation*}
	F(z) = \sum_{x\in\cone_{\delta}(t_0)} e^{(x,z)_d} P_0(\exists N:\ S_N=x).
\end{equation*} The cone-containment property of the steps together with a classical manipulation on generating functions (induced by the renewal property of \(P\)) implies that
\begin{equation*}
	F(z) = \frac{1}{1-H(z)},
\end{equation*}where
\begin{equation*}
	H(z) = \sum_{x\in\cone_{\delta}(t_0)} e^{(x,z)_d} \bfp(\displace(\gamma_1)=x).
\end{equation*}

So the the boundary of the convergence domain of \(F(z)\) is given by the set of \(z\) such that \(H(z)=1\) (we used that \(H(z)\) has a domain of convergence strictly larger than \(F(z)\)). It is thus the zero set of an analytic function and is thus an analytic manifold. It is worth noticing that the whole procedure is amounting to identify the boundary of the Wulff shape close to \(t_0\) with the rate function of a suitable random walk.

Finally, \(\psi_d\) is locally analytic in the direction as it is a convergent sum of locally analytic functions (see \cite{Campanino+Ioffe-2002} and \cite{Campanino+Ioffe+Velenik-2003} for details).
\section{Coarse-Graining Procedure and Cone Structure}
\label{sec:CoarseGraining}

Following the analysis of \cite{Campanino+Ioffe+Velenik-2003} and \cite{Campanino+Ioffe+Velenik-2008}, the first step of the proof of Theorem~\ref{thm:mainThm} is the approximation of \(C_{0,x}\) by a tree. It consists in two layers of coarse graining: the first is an approximation of the cluster in \(\bfm\) induced by the source constraint \(\partial\bfm=\{0,x\}\) and \(0\nleftrightarrow g\), it will give a tree; the second will be an approximation of the remainder of \(0\)'s cluster in \(\bfn+\bfm\), adding branches to the previously constructed tree. For this whole section, we consider the positive field double random current on \(\Zd\) with sources \(\varnothing\) and \(\{0,x\}\). By convention, the first current will be denoted \(\bfn\) and has \(\partial\bfn=\varnothing\) while the second will be denoted \(\bfm\) and has \(\partial\bfm=\{0,x\}\).

\subsection{The Coarse-Graining}

The coarse graining is a ``two layers'' version of the one done in~\cite{Campanino+Ioffe+Velenik-2008}: fix \(K>0\) large (the scale at which the analysis will be done), then, for \(y\in\Zd\), define
\begin{gather*}
\CGballK(y) = y+K\equiDecSet,\qquad \CGballKFat(y) = \CGball_{K+3\log( K)^3}(y).
\end{gather*}
For \(A\subset\Zd\),
\begin{gather*}
\bar{A}_K=\bigcup_{y\in A}\CGball_{3\log( K)^3}(y),\qquad [A]_K=\bigcup_{y\in A} \CGballKFat(y).
\end{gather*}

Recall that as \(\partial\bfm=\{0,x\}\), and \(0\nleftrightarrow g\), \(\bfm\) contains a long cluster. We start by coarse graining it using the following algorithm:

\begin{algorithm}[H]
	\label{alg:mainExtraction}
	Set \(v_0=0\), \(V=\{v_0\}\), \(n=1\)\;
	\While{\(A=\big\{ z\in\partial^{\exterior}[V]_K: z\xleftrightarrow{\CGballK(z)\setminus[V]_K }\partial^{\exterior}\CGballK(z)\big\}\neq \varnothing   \)}{
		Set \(v_{n}= \min A \)\;
		Update \(V=V\cup\{v_{n}\}\), \(n=n+1\)\;
	}
	Set \(M_1=n-1\)\;
	\Return \(V=\{v_0,...,v_{M_1}\}\)\;
	\caption{Coarse graining of the cluster of \(0\) in \(\bfm\)}.
\end{algorithm}
Beware that connections are in \(\bfm\)! Add then an edge between \(v_k\) and the element of \(\{v_i,i<k: v_k\in\partial^{\exterior}\CGballKFat(v_i)\}\) with smallest index to obtain a tree. From it, extract a trunk \(\trunk\) by taking the (unique) path between \(0\) and the \(v_i\) closest to \(x\). Denote the number of edges in \(\trunk\) by \(|\trunk|\).

Now, one want to add vertices to the previously discovered ones in order to have a coarse graining of \(0\)'s cluster in \(\bfn+\bfm\). Starting with the sequence \(V=\{v_0,...,v_{M_1}\}\):

\begin{algorithm}[H]
	\label{alg:noiseExtraction}
	Set \(n=1\), \(V=\{v_0,...,v_{M_1}\}\), \(V'=\varnothing\)\;
	\While{\(A=\big\{ z\in\partial^{\exterior}[V\cup V']_K: z\xleftrightarrow{\CGballK(z)\setminus[V]_K }\partial^{\exterior}\CGballK(z)\big\}\neq \varnothing   \)}{
		Set \(w_{n}= \min A \)\;
		Update \(V'=V'\cup\{w_{n}\}\), \(n=n+1\)\;
	}
	Set \(M_2=n-1\)\;
	\Return \(V'=\{w_1,...,w_{M_2}\}\)\;
	\caption{Coarse-graining of the remainder of the cluster in \(\bfn+\bfm\).}
\end{algorithm}
There the connections are understood in \(\bfn+\bfm\). Add again an edge between \(w_k\) and the smallest element of \(\{v_i: w_k\in\partial^{\exterior}\CGballKFat(v_i)\}\cup\{w_i,i<k: w_k\in\partial\CGballKFat(w_i)\}\) to obtain a tree. Call the vertices of \((V\cup V')\setminus\trunk \) the branches \(\branches\) of the tree and denote \(|\branches| \) the number of branches.

Denote \(\treeMap(\bfn,\bfm)\) the tree extracted from \((\bfn,\bfm)\). For a given tree \(\tree\) obtainable via the tree extraction, define its weight by:
\begin{equation}
	\label{eq:treeWeight}
	\treeWeight(\tree) = \sum_{\substack{\partial\bfn=\varnothing,\partial\bfm=\{0,x\}\\\treeMap(\bfn,\bfm)=\tree}} \frac{1}{\RCPF_{\Lambda_g}(\varnothing)^2}\weight(\bfn)\weight(\bfm).
\end{equation}

We stress at this point that the main difference between the tree created by the present coarse-graining and the tree obtained via the coarse graining of~\cite{Campanino+Ioffe+Velenik-2008} is that the vertices in \(\trunk\) imply a different event in the underlying model than the vertices in \(\branches\). Otherwise, the trees are the same (the only difference is that we use \(\CGballKFat = \CGball_{K+3\log( K)^3}\) instead of \(\CGballKFat = \CGball_{K+r\log( K)}\) for some \(r\geq 0\)). Results depending only on estimates on \(\treeWeight(\tree)\) are independent of the underlying model.

The next step is to give an estimate on the energy associated with a tree. It is the content of the next technical Lemma, the proof of which is postponed to Section~\ref{sec:TreeEnergyExtrac} as it contains most of the technical difficulties of this paper.
\begin{lemma}
	\label{lem:TreeEnergy}
	There exists \(\nu=\nu(h,d,\bfJ)>0\) such that:
	\begin{equation}
		\treeWeight(\tree) \leq e^{-K|\trunk|(1+o_{K}(1)) - \nu K|\branches|(1+o_K(1))}.
	\end{equation}
\end{lemma}

\subsection{Cone-Points of Trees}

Denote \(\ConePts_{\delta,t}(C)\) the set of cone-points of \(C\) for the cones \(\cone_{\delta}(t)\). Lemma~\ref{lem:TreeEnergy} is the input needed in~\cite[Section 2]{Campanino+Ioffe+Velenik-2008} to obtain:

\begin{theorem}
	\label{thm:CPofTree}
	Let \(\delta\in(0,1)\). Then, for any \(\epsilon>0\) there exist \(\nu>0\) and \(K_0\geq 0\) such that for any \(t\in\partial\WulffShape\) and \(x\in \Zd\) dual to \(t\) with \(\norm{x}\gg 1\),
	\begin{equation*}
		e^{\icl(x)}\treeWeight\Big( |\ConePts_{\delta,t}(\tree)|\leq (1-\epsilon)|\trunk| \Big)\leq e^{-\nu\norm{x}}.
	\end{equation*}
\end{theorem}
Cone-points of trees are cone-points of the tree vertices. The proofs are \emph{the same} as in~\cite{Campanino+Ioffe+Velenik-2003,Campanino+Ioffe+Velenik-2008} as they only involve the measure \(\treeWeight\) on trees and the only needed input is Lemma~\ref{lem:TreeEnergy}. The proofs are formulated there for \(\nu=1\) in Lemma~\ref{lem:TreeEnergy} and for \(\CGballKFat = \CGball_{K+r\log( K)}\) but they work with no modification to handle the case \(\nu>0\) and \(\CGballKFat = \CGball_{K(1+o_K(1))}\). The interested reader can find a model-independent presentation of these proofs in the author's PhD thesis~\cite{Ott-2019}. We sketch extremely briefly the procedure to highlight the key role played by Lemma~\ref{lem:TreeEnergy}.
\begin{itemize}
	\item First, notice that the entropy of trees with given trunk and branches sizes is negligible compared to their energy. Trees are thus governed by their energy.
	\item One then wants to test the energy of a tree against \(\icl(x)\), but the size of the trunk has to be, by construction of the tree, at least \(\icl(x)/K\). Thus, the minimal energy cost of a tree is \(\frac{\icl(x)}{K}K=\icl(x)\) and is given by a trunk with no branches and a straight trunk linking \(0\) to \(x\). We stress at this point that it is \emph{crucial} to have \(-K|\trunk|\) in Lemma~\ref{lem:TreeEnergy} and \emph{not} \(-\nu K|\trunk|\) for some \(\nu>0\).
	\item Via a careful analysis, one can then eliminate too large deviations from the straight line as they purely add to the energy.
	\item Branches are finally dealt with using a relatively simple energy-entropy argument (their presence add to the energy cost of the tree if \(K\) is chosen large enough).
\end{itemize}

\section{Cone-Points of Clusters and Factorization of Measure}
\label{sec:FEandFactorization}

In this section we turn the coarse control of the previous section into a fine control on microscopic clusters by a finite-energy type argument. One can then define irreducible pieces of the cluster as in \cite{Campanino+Ioffe+Velenik-2008} and use Appendix C of \cite{Ott2018} to give a factorized representation of the measure (see also Section 4 of \cite{Ott2018}).

\subsection{Cone Points of Clusters}

The goal of this section is the proof of
\begin{theorem}
	\label{thm:CPofCluster}
	There exist \(\delta\in(0,1)\) and \(\Cl{CP_density},\Cl{CP_rate}>0\) such that
	\begin{equation*}
		e^{\icl(x)}\RCLaw^{\varnothing,\{0,x\}}\Big( |\ConePts_{\delta,t}(C_{0,x})|\leq \Cr{CP_density}\norm{x}, 0\nleftrightarrow g \Big)\leq e^{-\Cr{CP_rate}\norm{x}},
	\end{equation*}
	uniformly over \(x\) and \(t\in\partial\WulffShape\) dual to \(x\).
\end{theorem}

Start by fixing \(x\) and \(t\in\partial\WulffShape\) dual to \(x\). For \(M>0\) define two families of slabs by
\begin{gather*}
	S_i=\Big\{z\in\R^d: 7(i-1)M \leq \big(z,\frac{x}{\norm{x}}\big)_d< 7iM\Big\},\\
	\tilde{S}_i=\Big\{z\in\R^d: 7(i-1)M +2M \leq \big(z,\frac{x}{\norm{x}}\big)_d< 7iM-2M\Big\},
\end{gather*} for \(i=1,...,\frac{\norm{x}}{7M}\). Then, a point \(z\in C_{0,x}\) will be said to be \((M,\delta)\)-good if
\begin{equation*}
	C_{0,x}\subset\Big( (\fcone_{\delta} +z)\cup (\bcone_{\delta} +z) \cup ([-M,M]^d+z) \Big).
\end{equation*} We will say that a slab \(S_i\) is \((M,\delta)\)-good if there exists \(z\in C_{0,x}\cap\tilde{S}_i\) such that \(z\) is \((M,\delta)\)-good. Denote \( \#_{\good}^{M,\delta} \) the number of \((M,\delta)\)-good slabs.
\begin{lemma}
	\label{lem:ClusterCPlem1}
	For every \(\delta\in(0,1)\), there exist \(M=M(\delta)\geq 0\) and \( \nu=\nu(\delta)>0\) such that
	\begin{equation*}
	e^{\icl(x)}\RCLaw^{\varnothing,\{0,x\}}\Big( \#_{\good}^{M,\delta} \leq \frac{\norm{x}}{14M} , 0\nleftrightarrow g\Big)\leq e^{-\nu\norm{x}},
	\end{equation*}
\end{lemma}
\begin{proof}
	First notice that for a given tree, increasing \(\delta\) can only increase the number of \((\delta,t)\)-cone-points. Then, let \(C=C(\delta)\) be such that every \((\delta/2,t)\)-cone-points of \(\tree\) is a \((CK,\delta)\)-good point (such a \(C\) exists as \(C_{0,x}\) is contained in a \(K+\log(K)^3\)-neighbourhood of \(\tree\)).
	\begin{figure}[h]
		\centering
		\includegraphics[scale=1]{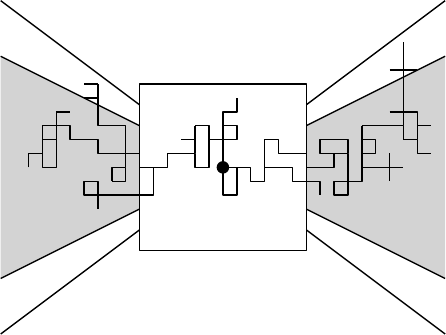}
		\label{fig:PreCP}
		\caption{Good point, \(\fcone_{\delta/2}\) is depicted in grey.}
	\end{figure}
	Let \(M=CK\). \(\#_{\good}^{M,\delta}\) is thus lower bounded by the number of \(\tilde{S}_i\) containing a \(\delta/2\)-cone-point. This implies that
	\begin{equation*}
		\Big\{\#_{\good}^{M,\delta} \leq \frac{\norm{x}}{14M}\Big\}\subset \Big\{|\ConePts_{\delta/2}(\tree)|\leq (1-\epsilon)\frac{\icl(x)}{K}\Big\}
	\end{equation*} for some \(\epsilon>0\). Now, by construction, \(|\trunk|\geq \frac{\icl(x)}{K+3\log(K)^3}\), so \(\frac{\icl(x)}{K}\leq (1+\frac{3\log(K)^3}{K})|\trunk|\). Choosing \(K\) such that \(\frac{3\log(K)^3}{K} \leq \frac{\epsilon}{2-2\epsilon}\),
	\begin{equation*}
	\Big\{|\ConePts_{\delta/2}(\tree)|\leq (1-\epsilon)\frac{\icl(x)}{K}\Big\}\subset \Big\{|\ConePts_{\delta/2}(\tree)|\leq (1-\epsilon/2)|\trunk|\Big\}.
	\end{equation*}
	Thus, Theorem~\ref{thm:CPofTree} gives:
	\begin{equation*}
		e^{\icl(x)}\treeWeight\Big( |\ConePts_{\delta/2}(\tree)|\leq (1-\epsilon/2)|\trunk| \Big)\leq e^{-\nu\norm{x}}.
	\end{equation*} This implies the wanted claim.
\end{proof}

Choose now \(\delta\) in such a way that \(\fcone_{\delta}(t)\) contains a unit coordinate vector. Let \(M\) be given by Lemma~\ref{lem:ClusterCPlem1} for \(\delta/2\). In particular, up to exponentially small error, half of the slabs are \((M,\delta/2)\)-good ones. Now, let \(I\) be a set of integers. If \(S_i, i\in I\) are good slabs, we describe a local surgery creating a \(\delta\)-cone-point in \(\tilde{S}_i\) for every \( i\in I\). We work with \(\bar{\bfn},\bar{\bfm}\) (the \(\{0,1,2\}\)-valued configurations giving the parity of the current) as they encapsulate all needed information.

\begin{itemize}
	\item For every \(i\in I\), let \(z_i\) be the smallest \((M,\delta/2)\)-good point in \(\tilde{S}_i\) (for the lexicographical order). Denote \(D=\bigcap_{i\in I}(z_i+\bcone_{\delta}\cup\fcone_{\delta})\).
	\item Set \(C=C_{0,x}(\bar{\bfn} + \bar{\bfm})\cap \Big(\bigcup_{i\in I}(z_i+[-M,M]^d)\Big)^c\). Denote \(C_{\bfn}\), resp. \(C_{\bfm}\), the open edges of \(C\) in \(\bar{\bfn}\), resp. \(\bar{\bfm}\).
	\item For all \(i\in I\), set all edges with an endpoint in \(z_i+[-M,M]^d\) that have no endpoint in \(C\) to \(0\) in both \(\bar{\bfn}\) and \(\bar{\bfm}\).
	\item For every source \(v\) created by the previous operation that is not a vertex of \(C\), adjust the value of \(\bar{\bfn}_{vg}\) and \(\bar{\bfm}_{vg}\) so that \(v\) is not a source any more (notice that this operation can not create a source at \(g\) as the number of sources not belonging to \(C\) created by the closure of edges is even).
	\item For each \(i\in I\), let \(B_i^{\blacktriangleleft}=(z_i+\fcone) \cap D \cap (z_i+[-M,M]^d) \) and \(B_i^{\blacktriangleright}=(z_i+\bcone) \cap D \cap (z_i+[-M,M]^d) \). Notice that \(B_i^{\blacktriangleleft}\cap B_i^{\blacktriangleright}=\{z_i\}\). Denote \(W_i^{\blacktriangleleft}\) the set of vertices in \(B_{i}^{\blacktriangleleft}\) that are connected to \(C\) (and define similarly \(W_i^{\blacktriangleright}\)). The graphs induced by \(\bfJ\) and \(B_i^{\blacktriangleleft},B_i^{\blacktriangleright}\) are connected by the assumption \(J_{v,v\pm\bfe_k}>0\).
	\item For every \(i\in I\) do: let \(\mathfrak{N}_i^{\blacktriangleleft}\) be the set of sources of \(\bfn\) in \(B_i^{\blacktriangleleft}\) (it has even cardinality). Let \(\bar{\bfa}_i^{\blacktriangleleft}\) be a deterministic parity function on the graph \((B_{i}^{\blacktriangleleft},\bfJ)\) with:
	\begin{itemize}
		\item \(\partial \bar{\bfa}_i^{\blacktriangleleft} = \mathfrak{N}_i^{\blacktriangleleft}\),
		\item \(\bar{\bfa}_i^{\blacktriangleleft}\) contains only one cluster with at least two vertices and \(W_i^{\blacktriangleleft}\) is in that cluster.
	\end{itemize}
	Such a parity function exists by the assumption \(J_{v,v\pm\bfe_k}>0\), Lemma~\ref{lem:sourceGraph} and the insertion tolerance property (Lemma~\ref{lem:RCInsTol}). Set \(\bar{\bfn}\) to be equal to \(\bar{\bfa}_i^{\blacktriangleleft}\) on \((B_{i}^{\blacktriangleleft},\bfJ)\). Do the same for \(\blacktriangleright\).
	\item For every \(i\in I\) do: let \(\mathfrak{M}_i^{\blacktriangleleft}\) be the set of sources of \(\bfm\) in \(B_i^{\blacktriangleleft}\) (it has odd cardinality). Let \(\bar{\bfd}_i^{\blacktriangleleft}\) be a deterministic parity function on the graph \((B_{i}^{\blacktriangleleft},\bfJ)\) with:
	\begin{itemize}
		\item \(\partial \bar{\bfd}_i^{\blacktriangleleft} = \mathfrak{M}_i^{\blacktriangleleft}\cup\{z_i\}\),
		\item \(\bar{\bfd}_i^{\blacktriangleleft}\) contains only one cluster with at least two vertices and \(W_i^{\blacktriangleleft}\) is in that cluster.
	\end{itemize}
	Such a parity function again exists for the same reasons as in the previous point. Set \(\bar{\bfm}\) to be equal to \(\bar{\bfd}_i^{\blacktriangleleft}\) on \((B_{i}^{\blacktriangleleft},\bfJ)\). Do the same for \(\blacktriangleright\).
\end{itemize}
Notice that when opening edges in the last two points, one does not create connexions to the ghost as the opened edges belong to \((\Zd,\bfJ)\) and not to the graph augmented with a ghost. The obtained configurations satisfy \(\partial\bar{\bfn}=\varnothing,\partial\bar{\bfm}=\{0,x\}\), \(0\nleftrightarrow g\) and all \(z_i\) are cone-points. Denote \(Y^I\) the above surgery.

\begin{figure}[h]
	\centering
	\includegraphics[scale=0.8]{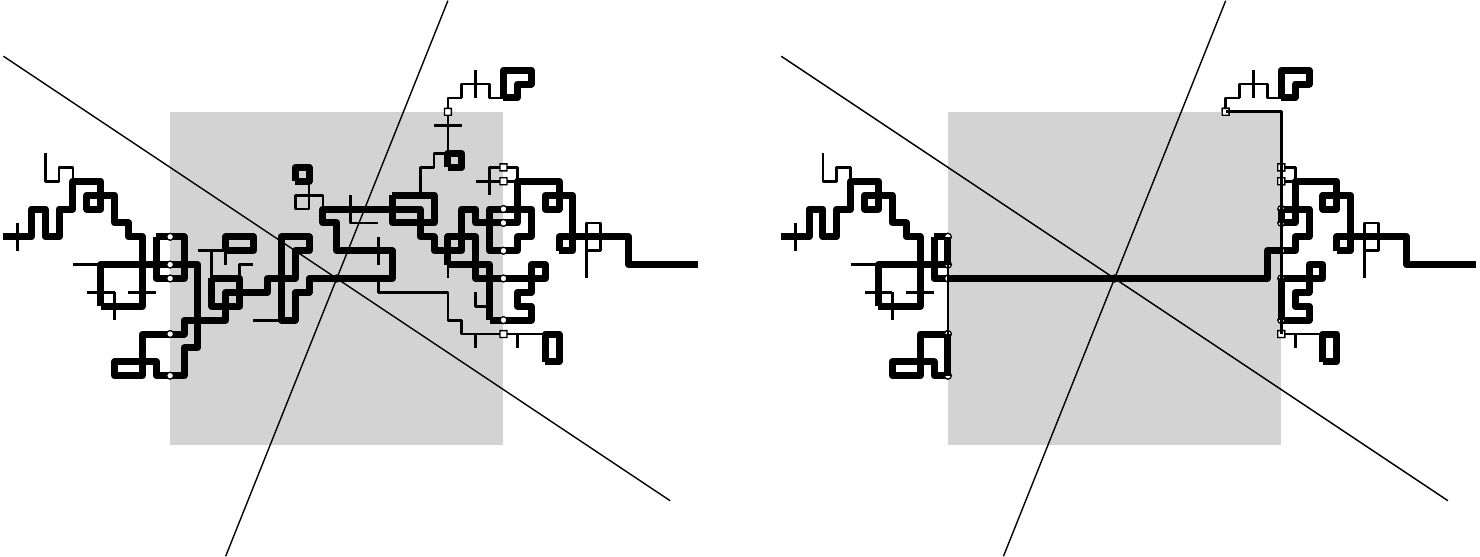}
	\caption{Local surgery of \(C_{0,x}\) in \(\bfm\). Odd edges are represented by fat lines and even ones by thinner ones.}
	\label{fig:PreCPtoCP}
\end{figure}

The next lemma will conclude the proof of Theorem~\ref{thm:CPofCluster}.
\begin{lemma}
	\label{lem:ClusterCPlem2}
	There exist \(\rho>0,\nu>0\) such that
	\begin{equation*}
		e^{\icl(x)}\RCLaw^{\varnothing,\{0,x\}}\Big( |\ConePts_{\delta}(C_{0,x})|\leq \rho\frac{\norm{x}}{14M}, \#_{\good}^{M,\delta/2} \geq \frac{\norm{x}}{14M} , 0\nleftrightarrow g\Big)\leq e^{-\nu\norm{x}}.
	\end{equation*}
\end{lemma}
\begin{proof}
	We work in finite volume and will take limits at the end of the proof. Let \(\epsilon\) be small (to be chosen later). First notice that,
	\begin{equation*}
		\RCLaw^{\varnothing,\{0,x\}}_{\Lambda_g}\big( \bfn,\bfm \big)\leq e^{ \tilde{c}M^d|I|}\RCLaw^{\varnothing,\{0,x\}}_{\Lambda_g}\big( Y^I(\bfn,\bfm) \big),
	\end{equation*}
	(where \(Y^I\) is the surgery defined previously). Then, remark that when transforming one of the \(z_i\) into a cone-point, one might actually create more than one cone-point but no more than \(C_{M}=C_M(M,\delta,d)>0\) as the created cone-points have to be \(\delta/2\)-seen by \(z_i\) and have to be \(\delta\)-blocked by \(z_i+[-M,M]^d\) (see Lemma~\ref{lem:createdCPvolume}). Now, let \((\bfn,\bfm)\) be a configuration with at least \( \frac{\norm{x}}{14M}\) \(\delta/2\)-good slabs and with \(C_{0,x}\) containing less than \(\epsilon\frac{\norm{x}}{14M}\) cone-points. We do a many-to-many argument:
	\begin{itemize}
		\item Denote \(G\) the set of good slabs of \((\bfn,\bfm)\).
		\item Define \[Y(\bfn,\bfm)=\bigcup_{\substack{I\subset G: |I|=\epsilon\frac{\norm{x}}{14M}}} Y^I(\bfn,\bfm).\]
		\item Denote \(Y^{-1}(\bfn',\bfm')=\big\{ (\bfn,\bfm): (\bfn',\bfm')\in Y(\bfn,\bfm)  \big\}\).
	\end{itemize}
	\(Y\) is a multivalued map from the set of configuration with less than \(\epsilon\frac{\norm{x}}{14M}\) \(\delta\)-cone-points and at least \( \frac{\norm{x}}{14M}\) \(\delta/2\)-good slabs, denoted \(A\), to the set of configuration with at least \(\epsilon\frac{\norm{x}}{14M}\) \(\delta\)-cone-points, denoted \(B\) (the event \(0\nleftrightarrow g\) is implicit here). Then,
	\begin{align*}
		&\RCLaw^{\varnothing,\{0,x\}}\Big( |\ConePts(C_{0,x})|\leq \epsilon\frac{\norm{x}}{14M}, \#_{\good}^{M,\delta} \geq \frac{\norm{x}}{14M} , 0\nleftrightarrow g\Big) = \\
		\ &=\sum_{(\bfn,\bfm)\in A} \frac{1}{|Y(\bfn,\bfm)|} \sum_{(\bfn',\bfm')\in Y(\bfn,\bfm)} \RCLaw^{\varnothing,\{0,x\}}\big( \bfn,\bfm\big).
	\end{align*}
	\(|Y(\bfn,\bfm)|\) is the number of choices of \(I\subset G\) with cardinality \(\epsilon\frac{\norm{x}}{14M}\), thus
	\begin{equation*}
		|Y(\bfn,\bfm)| = \binom{|G|}{\epsilon\frac{\norm{x}}{14M}} \geq \binom{\frac{\norm{x}}{14M}}{\epsilon\frac{\norm{x}}{14M}} \geq e^{\epsilon\log(1/\epsilon) \frac{\norm{x}}{14M}}.
	\end{equation*}
	In the same idea, to obtain \(Y^{-1}(\bfn',\bfm')\), one has to chose the \(\epsilon\frac{\norm{x}}{14M}\) slabs that were modified amongst the at most \(C_M\epsilon\frac{\norm{x}}{14M}\) slabs containing a \(\delta\)-cone-point, and then to look at what where the possible local configurations; the number of edges modified by \(Y^I\) is at most \(C\epsilon\frac{\norm{x}}{14M}M^d\) so for a given choice of modified slabs, one has at most \( \exp{c'\epsilon\frac{\norm{x}}{14M}M^d} \) possible ``pre-images''. So,
	\begin{equation*}
		|Y^{-1}(\bfn',\bfm')| \leq \binom{C_M\epsilon\frac{\norm{x}}{14M}}{\epsilon\frac{\norm{x}}{14M}}e^{c'\epsilon\frac{\norm{x}}{14M}M^d} \leq e^{ (c_M + c' M^d)\epsilon\frac{\norm{x}}{14M} }.
	\end{equation*}
	Putting everything together,
	\begin{align*}
		&\sum_{(\bfn,\bfm)\in A} \frac{1}{|Y(\bfn,\bfm)|} \sum_{(\bfn',\bfm')\in Y(\bfn,\bfm)} \RCLaw^{\varnothing,\{0,x\}}\big( \bfn,\bfm\big) \leq \\
		\ & \leq e^{-\epsilon\log(1/\epsilon) \frac{\norm{x}}{14M}} \sum_{(\bfn,\bfm)\in A} \sum_{(\bfn',\bfm')\in B}\mathds{1}_{(\bfn',\bfm')\in Y(\bfn,\bfm)} e^{\tilde{c}M^d\epsilon\frac{\norm{x}}{14M}}\RCLaw^{\varnothing,\{0,x\}}\big( \bfn',\bfm'\big)\\
		\ & \leq e^{(\tilde{c}M^d-\log(1/\epsilon)) \epsilon\frac{\norm{x}}{14M}}  \sum_{(\bfn',\bfm')\in B}\RCLaw^{\varnothing,\{0,x\}}\big( \bfn',\bfm'\big) \sum_{(\bfn,\bfm)\in A}\mathds{1}_{(\bfn',\bfm')\in Y(\bfn,\bfm)}\\
		\ & \leq e^{(\tilde{c}M^d-\log(1/\epsilon)) \epsilon\frac{\norm{x}}{14M}}  \sum_{(\bfn',\bfm')\in B}\RCLaw^{\varnothing,\{0,x\}}\big( \bfn',\bfm'\big) e^{ (c_M + c' M^d)\epsilon\frac{\norm{x}}{14M} }\\
		\ & = e^{(c_M + c' M^d+\tilde{c}M^d-\log(1/\epsilon)) \epsilon\frac{\norm{x}}{14M}}  \RCLaw^{\varnothing,\{0,x\}}\big( B\big).
	\end{align*}
	Choosing \(\epsilon<\exp{-c_M -(c'+\tilde{c})M^d}\), one gets \[\RCLaw^{\varnothing,\{0,x\}}\big( A\big)\leq e^{-\nu_{\epsilon} \norm{x}} \RCLaw^{\varnothing,\{0,x\}}\big( B\big). \] Multiplying by \(e^{\icl(x)}\) on both sides and using \(B\subset \{0\nleftrightarrow g\}\), one gets the wanted estimate with \(\rho=\epsilon\).
\end{proof}

\subsection{Irreducible Structure}

We are now in position to describe the representation of \(C_{0,x}\) in irreducible pieces.

Fix now (and until the end of Section~\ref{sec:FEandFactorization}) any \(t_0\in\partial\WulffShape\) and choose any \(\delta\in(0,1)\) so that \(\fcone_{\delta}(t_0)\) contains a coordinate axis. We will denote this cone \(\fcone\).

It will be convenient to work with only one probability measure for all \(x\) dual to \(t\) rather than the family \(\RCLaw^{\varnothing,\{0,x\}}\). We can do it as follows: define \(Y\) the map from \(\{\bfm:\ \partial\bfm = \{0,x\}, 0\nleftrightarrow g \text{ in } \bfm \} \) that sets \(\bfm_{0g}=\bfm_{xg}=1\) and does not change anything else. It is a bijection between its domain and its image and we have \(\partial Y(\bfm)=\varnothing\) and \(\weight(Y(\bfm))=h^2 \weight(\bfm)\) whenever we look at a finite graph. Denote \(\Gamma\) the graph obtained from \((\Zd\cup\{g\},\tilde{\bfJ})\) by removing the edges \(\{0,g\}\) and \(\{x,g\}\). Using this (taking the wanted quantity in finite volume and taking limits), we get that for any event \(A\) measurable with respect to the value of the edges in \(\Gamma\):
\begin{equation*}
	\langle \sigma_0\sigma_x \rangle \RCLaw^{\varnothing,\{0,x\}}(0\nleftrightarrow g, A) = h^{-2}\RCLaw^{\varnothing,\varnothing}(0\stackrel{\Gamma}{\nlongleftrightarrow} g, A, \bfm_{0g}=\bfm_{xg}=1).
\end{equation*}

Let \(\gamma\) be an irreducible cluster (we will understand \(\partial^{\edge}\gamma\) as being the edge boundary of \(\gamma\) in \((\Zd,\bfJ)\)), define the event ``\(\gamma\) is an irreducible piece'', \(A(\gamma)\) by:

\begin{enumerate}
	\item \((\bfn+\bfm)_e >0\) for all \(e\in \gamma\),
	\item \((\bfn+\bfm)_e =0\) for all edge \(e=\{v,g\}\) with \(v\in\gamma\setminus\{\bend(\gamma)\}\),
	\item \((\bfn+\bfm)_e =0\) for all \(e\in \partial^{\edge}(\gamma\setminus \{\fend(\gamma),\bend(\gamma)\})\),
	\item \((\bfn+\bfm)_e =0\) for all edge \(e=\{\fend(\gamma),v\}\) not in \(\gamma\) with \(v\notin \fend(\gamma)+\bcone\),
	\item \((\bfn+\bfm)_e =0\) for all edge \(e=\{\bend(\gamma),v\}\) not in \(\gamma\) with \(v\in \bend(\gamma)+\bcone\),
	\item \(\partial\bfn|_{\gamma} = \varnothing\), \(\partial\bfm|_{\gamma} = \{\fend(\gamma),\bend(\gamma)\}\) with \(\bfn|_{\gamma}\) the restriction of \(\bfn\) to the edges of \(\gamma\).
\end{enumerate}

In the same fashion, for \(\gamma_L\) a backward irreducible marked cluster, define \(A(\gamma_L)\) by
\begin{enumerate}
	\item \((\bfn+\bfm)_e >0\) for all \(e\in \gamma_L\),
	\item \((\bfn+\bfm)_e =0\) for all edge \(e=\{v,g\}\) with \(v\in\gamma\setminus\{\bend(\gamma), v^*\}\),
	\item \(\bfm_{v^* g} =1\),
	\item \((\bfn+\bfm)_e =0\) for all \(e\in \partial^{\edge}(\gamma\setminus \{\bend(\gamma)\})\),
	\item \((\bfn+\bfm)_e =0\) for all edge \(e=\{\bend(\gamma),v\}\) not in \(\gamma\) with \(v\in \bend(\gamma)+\bcone\),
	\item \(\partial\bfn|_{\gamma} = \varnothing\), \(\partial\bfm|_{\gamma} = \{v^*,\bend(\gamma)\}\) with \(\bfn|_{\gamma}\) the restriction of \(\bfn\) to the edges of \(\gamma\).
\end{enumerate}

\begin{figure}[h]
\centering
\includegraphics[scale=0.9]{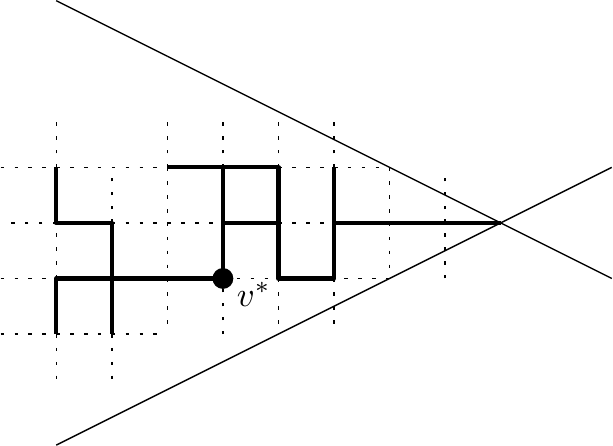}
\caption{The event \(A(\gamma_L)\). Fat edges have to be \(>0\), dotted ones \(=0\). The mark represent an open edge connecting to \(g\). All other edges connecting to \(g\) are closed.}
\label{fig:EventBackIrred}
\end{figure}
And define similarly \(A(\gamma_{R})\) ``\(\gamma_{R}\) is a forward irreducible marked piece''. For a chain of irreducible pieces \(\gamma=\gamma_L\sqcup \gamma_1\sqcup...\sqcup\gamma_{M}\sqcup \gamma_R\) with \(D(\gamma)=x\), \(x\) dual to \(t_0\), one has (with \(\gamma_0=\gamma_L,\gamma_{R}=\gamma_{M+1}\))
\begin{multline}
	e^{\icl(x)} \langle \sigma_0\sigma_x \rangle h^2 \RCLaw^{\varnothing,\{0,x\}} ( C_{0,x}=\gamma, 0\nleftrightarrow g) =\\ = e^{(t_0,x)_d} \RCLaw^{\varnothing,\varnothing} ( A(\gamma_0))\prod_{k=1}^{M+1} \RCLaw^{\varnothing,\varnothing} ( A(\gamma_k)|A(\gamma_0),...A(\gamma_{k-1}))\\
	= \procMem(\gamma_0)\prod_{k=1}^{M+1} \procMem(\gamma_k|\gamma_0,...\gamma_{k-1})
\end{multline}
where \(\procMem(\gamma_k|\gamma_0,...\gamma_{k-1}) = e^{(t_0,D(\gamma_k))_d}\RCLaw^{\varnothing,\varnothing} ( A(\gamma_k)|A(\gamma_0),...A(\gamma_{k-1})) \). 

Notice that Theorem~\ref{thm:CPofCluster} (or more precisely, its proof), implies that \(\procMem(\norm{D(\gamma_k)}\geq l|\gamma_0,...\gamma_{k-1}) \leq e^{-cl}\) for some \(c>0\), uniformly in \(\gamma_0,...\gamma_{k-1}\) as this event implies a long cluster with no cone-point. It also gives that \(M\geq \rho\norm{x}\) for some \(\rho>0\) up to an error of order \(e^{-c' \norm{x}}\). See \cite{Campanino+Ioffe+Velenik-2003} for more details.

\subsection{Factorization of Measure}

We are at the stage where the analysis done in~\cite{Campanino+Ioffe+Velenik-2003,Campanino+Ioffe+Velenik-2008} goes through Ruelle operator analysis. We follow here the different road, paved in~\cite[Section 4]{Ott2018}, that allows a complete factorization of the weights.

We are exactly in the set-up of \cite[Appendix C]{Ott2018} with the alphabet being \(\SetRootDiaIrr\), the boundary conditions being \(\SetRootMarkBackIrr\) and \(\SetRootMarkForwIrr\) and the kernel given by \(\procMem\). Theorem~\ref{thm:mainThm} is thus a direct application of Lemma C.1 and Theorem C.4 of \cite{Ott2018}. We now prove the relevant properties of the process with memory \(\procMem\) in order to be able to apply them, finishing the proof of Theorem~\ref{thm:mainThm}.

\begin{lemma}
	Hypotheses (H1) to (H4) and properties (P1) to (P5) of Appendix C in \cite{Ott2018} are fulfilled for the process \(\procMem\) described in the previous section.
\end{lemma}
\begin{proof}
	(H1) is by construction of the process (choice of weights). Exponential ratio mixing, (H2), is Theorem~\ref{thm:expMixRC}, sub-exponential growth of the mass, (H3), is by the normalization of the weight (see the argument in~\cite[Section 4]{Ott2018}). (H4) as well as (P3) and (P4) is the ``finite energy'' used in the proof of Lemma~\ref{lem:ClusterCPlem2}. (P1) is a direct consequence of Theorem~\ref{thm:CPofCluster}, (P2) is the cone-containment property and (P5) is implied by the symmetries of \(\mu\) induced by the symmetries of the interaction \(\bfJ\).
\end{proof}

Applying Lemma C.1 and Theorem C.4 of~\cite{Ott2018} to \(\procMem\), one obtains \(\rho_L,\rho_R\) two finite measures on \(\SetRootMarkBackCont\) and \(\SetRootMarkForwCont\) respectively and a probability measure \(\bfp\) on \(\SetRootDiaCont\) such that:

\begin{itemize}
	\item There exists \(\nu>0\) with
	\begin{equation*}
	\rho_L(\norm{D(\gamma_L)}\geq l) \vee \rho_R(\norm{D(\gamma_R)}\geq l) \vee \bfp(\norm{D(\gamma_1)}\geq l) \leq e^{-\nu l}.
	\end{equation*}
	\item There exists some \(c>0\) such that, for any \(x\) dual to \(t_0\) and \(f\) function of the joint cluster of \(0\) and \(x\),
	\begin{multline*}
	\Big|e^{\icl(x)}\sum_{C} \RCLaw^{\varnothing,\varnothing}(0\stackrel{\Gamma}{\nlongleftrightarrow} g, C_{0,x}|_{\Zd}=C, \bfm_{0g}=\bfm_{xg}=1) f(C) - \\
	-\sum_{\gamma_L}\sum_{\gamma_R}\rho_L(\gamma_L)\rho_{R}(\gamma_R)\sum_{M\geq 0}\sum_{\gamma_1,...,\gamma_M}\prod_{i=1}^{M} \bfp(\gamma_i) \mathds{1}_{D(\gamma)=x}  f(\gamma)\Big|\leq \norm{f}_{\infty} e^{-c\norm{x}},
	\end{multline*}where \(\gamma=\gamma_L\concatenate\gamma_1\concatenate\cdots \concatenate\gamma_M\concatenate\gamma_R\) and the sums are over \(\gamma_L\in\SetRootMarkBackCont\), \(\gamma_R\in \SetRootMarkForwCont\) and \(\gamma_i\in \SetRootDiaCont\) (the boundary words of~\cite{Ott2018} are \(\SetRootMarkBackCont\) and \(\SetRootMarkForwCont\) and the set of words is \(\SetRootDiaCont\)).
\end{itemize}

Choose now \(t\in\partial\WulffShape\) close to \(t_0\). Keep the same set of irreducible pieces and the same events \(A(\gamma_{*}), *=i,L,R\). We can define a ``new'' kernel:
\begin{equation*}
	\tilde{\procMem}(\gamma_k|\gamma_0,... \gamma_{k-1}) = e^{(t,D(\gamma_k))_d}\RCLaw^{\varnothing,\varnothing} ( A(\gamma_k)|A(\gamma_0),..., A(\gamma_{k-1})) .
\end{equation*}
We have \(\frac{\tilde{\procMem}(\gamma_k|\gamma_0,..., \gamma_{k-1})}{\procMem(\gamma_k|\gamma_0,..., \gamma_{k-1})} = e^{(t-t_0,\displace(\gamma_k))}\) for any sequence of irreducible pieces. As for \(\procMem\), one can construct \(\tilde{\rho}_L,\tilde{\rho}_R,\tilde{\bfp}\) on \(\SetRootMarkBackCont,\SetRootMarkForwCont,\SetRootDiaCont\). Checking the formula for \(\rho_L,\rho_R,\bfp\), \cite[eq. (32)]{Ott2018}, one has that \(\frac{\tilde{\bfp}(\gamma_1)}{\bfp(\gamma_1)} = e^{(t-t_0,\displace(\gamma_1))}\) for any \(\gamma_1\in\SetRootDiaCont\) (and similarly for \(\rho_{L/R}\)).

So, for any \(x\) dual to \(t\) and any function \(f\) of the joint cluster of \(0\) and \(x\),
\begin{multline*}
\Big|e^{(t_0,x)}\sum_{C} \RCLaw^{\varnothing,\varnothing}(0\stackrel{\Gamma}{\nlongleftrightarrow} g, C_{0,x}|_{\Zd}=C, \bfm_{0g}=\bfm_{xg}=1) f(C) - \\
-\sum_{\gamma_L}\sum_{\gamma_R}\rho_L(\gamma_L)\rho_{R}(\gamma_R)\sum_{M\geq 0}\sum_{\gamma_1,...,\gamma_M}\prod_{i=1}^{M} \bfp(\gamma_i) \mathds{1}_{D(\gamma)=x}  f(\gamma)\Big| \\\leq  \norm{f}_{\infty} e^{-c\norm{x}}e^{(t_0-t, x)}.
\end{multline*}
Taking \(\norm{t-t_0}\leq c/2\) we obtain Theorem~\ref{thm:mainThm}.

\section{Tree Energy Extraction (Proof of Lemma~\ref{lem:TreeEnergy})}
\label{sec:TreeEnergyExtrac}

As the proof of Lemma~\ref{lem:TreeEnergy} is quite long, we divide it into several steps: first we approximate the quantity to estimate by local quantities, then we use exponential mixing in the random current to factorize the local quantities and we finish by proving the relevant bounds for the factorized quantities.

\subsection{Approximation by Local Events and Factorization}

Denote \(v_0,... ,v_{|\trunk|-1}\) the vertices of \(\trunk\) in the order of discovery by Algorithm~\ref{alg:mainExtraction} and \(b_1,... ,b_{|\branches|}\) the branches in order of discovery by successive application of Algorithm~\ref{alg:mainExtraction} and Algorithm~\ref{alg:noiseExtraction}. Denote
\begin{gather*}
	\Delta(v_k)= \CGballK(v_k)\setminus\bigcup_{i=0}^{k-1}\CGballKFat(v_i), \quad \bar{\Delta}(v_k)=\bigcup_{z\in\Delta(v_k)}\CGball_{\log( K)^3}(z)\\
	\Delta(b_k)= \CGballK(b_k)\setminus\Big([\trunk]_K\cup\bigcup_{i=0}^{k-1}\CGballKFat(b_i)\Big).
\end{gather*}

\begin{figure}[h]
	\centering
	\includegraphics[scale=0.65]{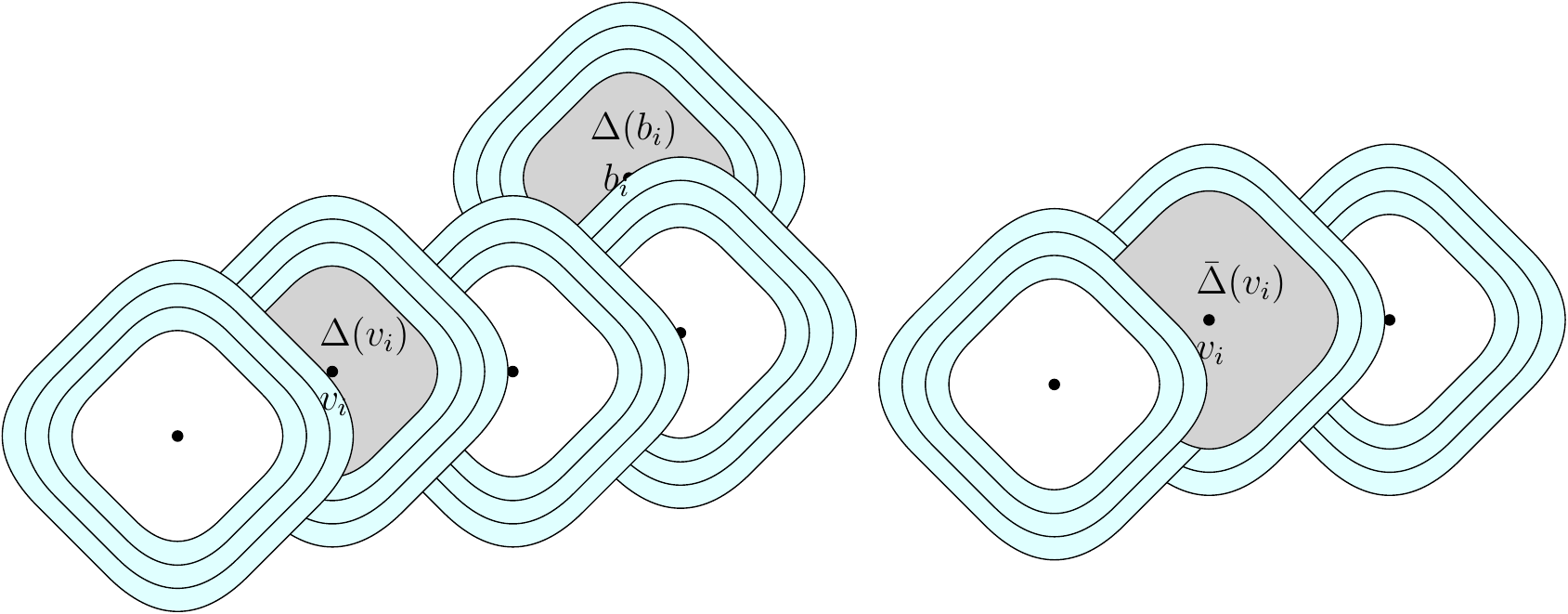}
	\label{fig:LocalNeigh}
\end{figure}

Then define the (local) events:
\begin{gather*}
	F_k= \{v_k\xleftrightarrow{\Delta(v_k)}\partial^{\exterior}\CGballK(v_k) \text{ in } \bfm \}\cap\{ v_k\stackrel{\bar{\Delta}_g(v_k)}{\nlongleftrightarrow} g \text{ in } \bfm+\bfn\},\\
	H_k = \{b_k\xleftrightarrow{\Delta(b_k)}\partial^{\exterior}\Delta(b_k) \text{ in } \bfm+\bfn\}\cap\{ b_k\stackrel{\Delta_g(b_k)}{\nlongleftrightarrow} g \text{ in } \bfm+\bfn \}.
\end{gather*}

Notice that the distance in \((\Zd,\bfJ)\) between the support of any of those event and the support of the others is at least \(c\log(K)^3\) for some \(c>0\) depending on \(d,\bfJ\).

Recalling the weight of a tree~\eqref{eq:treeWeight},
\begin{align}
	\label{eq:factorization}
	\nonumber \treeWeight(\tree) &= \sum_{\substack{\partial\bfn=\varnothing,\partial\bfm=\{0,x\}\\\treeMap(\bfn,\bfm)=\tree}} \frac{1}{\RCPF_{\Lambda_g}(\varnothing)^2}\weight(\bfn)\weight(\bfm)\\ \nonumber
	&= \langle\sigma_0\sigma_x\rangle_{\Lambda} \RCLaw^{\varnothing,\{0,x\}}_{\Lambda_g}\big(\treeMap(\bfn,\bfm)=\tree\big)\\ \nonumber
	&\leq \RCLaw^{\varnothing,\{0,x\}}_{\Lambda_g}\Big( \bigcap_{k=0}^{|\trunk|-1} F_k \cap \bigcap_{k=1}^{|\branches|} H_k \Big)\\
	&\leq \prod_{k=1}^{|\trunk|-1} e^{e^{-\frac{c}{2}\log(K)^3}}\RCLaw^{\varnothing,\varnothing}_{\Lambda_g}\big( F_k \big)\prod_{k=1}^{|\branches|}e^{e^{-\frac{c}{2}\log(K)^3}} \RCLaw^{\varnothing,\varnothing}_{\Lambda_g}\big( H_k \big)
\end{align}
where the last inequality, valid for \(K\) large enough, is exponential mixing (Theorem~\ref{thm:expMixRC}) and Lemma~\ref{lem:sourceToSourceless}.

In order to finish the proof of Lemma~\ref{lem:TreeEnergy}, it remains to show that there exists \(\nu>0\) such that:
\begin{gather*}
	\RCLaw^{\varnothing,\varnothing}_{\Lambda_g}\big( F_k \big) \leq e^{-K(1+o_K(1))},\\
	\RCLaw^{\varnothing,\varnothing}_{\Lambda_g}\big( H_k \big) \leq e^{-\nu K(1+o_K(1))}.
\end{gather*}
This is the object of the next section: the first inequality is Lemma~\ref{lem:trunkLocalEnergy} plus a union bound and the second is Lemma~\ref{lem:brancheLocalEnergy} plus a union bound.

\subsection{Energy Bounds on Local Quantities}

We start by the easiest bound:
\begin{lemma}
	\label{lem:brancheLocalEnergy}
	Let \(\Lambda=[-N,N]^d\) be a large square box in \(\Z^d\). Then, there exist \(\nu>0, K_0\geq 0\) (not depending on \(N\)) such that for any \(K\geq K_0\), any \(\Delta\subset\Lambda\) and any \(u,v\in\Delta\) with \(v\in \partial^{\exterior}\CGballK(u)\),
	\begin{equation*}
	\RCLaw^{\varnothing,\varnothing}_{\Lambda_g}\big( u\xleftrightarrow{\Delta} v, u\stackrel{\Delta_g}{\nlongleftrightarrow} g \big)\leq e^{-\nu K(1+o_K(1))}
	\end{equation*}
\end{lemma}
\begin{proof}
	\(u\xleftrightarrow{\Delta} v\) imply the existence of a path \(\gamma\) going from \(u\) to \(v\) in \(\Delta\). Using that this path contains at least \(cK\) vertices and the insertion tolerance property of Lemma~\ref{lem:RCInsTol}, one gets the wanted estimate.
\end{proof}

The next Lemma is much more technical: one has to extract locally a truncated correlation function which requires more precise control.

\begin{lemma}
	\label{lem:trunkLocalEnergy}
	Let \(\Lambda=[-N,N]^d\) be a large square box in \(\Z^d\). Then, there exists \(K_0\geq 0\) (not depending on \(N\)) such that for any \(K\geq K_0\), for any \(\Delta\subset \bar{\Delta}\subset\Lambda\) with \(d(\Delta,\partial^{\interior}\bar{\Delta})\geq \log(K)^3 \) and any \(u,v\in\Delta\) with \(v\in \partial^{\exterior}\CGballK(u)\),
	\begin{equation*}
	\RCLaw^{\varnothing,\varnothing}_{\Lambda_g}\big( u\stackrel{\Delta}{\leftrightarrow} v \text{ in } \bfm, u\stackrel{\bar{\Delta}_g }{\nlongleftrightarrow} g \text{ in } \bfn+\bfm \big)\leq e^{-K(1+o_K(1))}
	\end{equation*}
\end{lemma}
\begin{proof}
	Define:
	\begin{itemize}
		\item \(\calC_g\) to be the cluster (set of sites) of \(g\) in \(\bar{\Delta}_g\),
		\item \(\calC_{\partial}\) the sites of \(\bar{\Delta}\) connected to \(\partial^{\interior}\bar{\Delta}\) in \(\bar{\Delta
		}\) but not connected to \(g\) in \(\bar{\Delta}_g\),
		\item \(\calC_{\free}\) the sites of \(\bar{\Delta}\) that are neither connected to \(\partial^{\interior}\bar{\Delta}\) nor to \(g\) in \(\bar{\Delta}_g\).
	\end{itemize}
	\begin{figure}[h]
		\centering
		\includegraphics[scale=0.9]{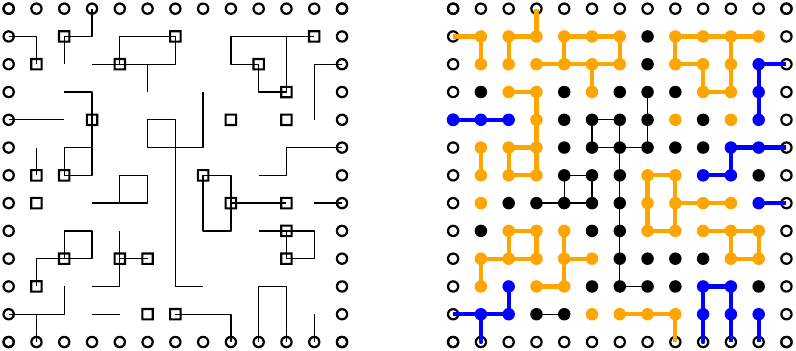}
		\caption{Left: edge configuration, squares represent edges connecting to the ghost, circle are the boundary sites. Right: \(\calC_g\) in orange, \(\calC_{\partial}\) in blue and \(\calC_{\free}\) in black.}
		\label{fig:ClusterType}
	\end{figure}
	One has \(\calC_g\sqcup \calC_{\partial}\sqcup \calC_{\free} = V_{\bar{\Delta}_g}\). At least one of the two following events occurs:
	\begin{enumerate}
		\item \(\calC_{\partial}\) contains at least \(\frac{K}{\log(K)^2}\) disjoint paths from \(\Delta\) to \(\partial^{\interior}\bar{\Delta}\). Call this event \(\calA_1\).
		\item \(\calC_{\partial}\) contains at most \(\frac{K}{\log(K)^2}\) disjoint paths from \(\Delta\) to \(\partial^{\interior}\bar{\Delta}\). Call this event \(\calA_2\).
	\end{enumerate}
	Notice that all events and objects considered here depend only on \(\bar{\bfn},\bar{\bfm}\).
	We then treat those two cases separately.
	\begin{case}
		\begin{multline*}
		\RCLaw_{\Lambda_g}^{\varnothing,\varnothing}\big( u\stackrel{\Delta}{\leftrightarrow} v \text{ in }\bfm,u\stackrel{\bar{\Delta} }{\nlongleftrightarrow} g \text{ in }\bfm+\bfn, \calA_1 \big) \leq\\ \RCLaw_{\Lambda_g}^{\varnothing,\varnothing}\big( \calA_1 \big)
		\leq e^{-c K\log(K)}\leq  e^{-K}
		\end{multline*}
		where the second inequality is finite energy: each of the paths picks up a weight \(e^{-c'}\) per site visited, see the proof of Lemma~\ref{lem:RCInsTol} (\(c\) depends on \(h\) and on the range of the interaction) and the last is \(K\geq K_0(c)\).
	\end{case}
	\begin{case}
		For a current \(\bfn\) with \(u\leftrightarrow v\) and \( u,v\in C_{\free}\), define the outermost closed contour surrounding \(u,v\) to be the contour \(\eta\) defined by:
		\begin{enumerate}
			\item explore \(C_{\partial}\) and \(C_{g}\), this determines \(C_{\free}\) from outside,
			\item let \(C_{u,v}\) be the connected component of \(u,v\) in the graph \((C_{\free},\bfJ)\),
			\item define \(\eta = \partial^{\edge} C_{u,v}\).
		\end{enumerate}
		Notice that \(u\stackrel{\Delta}{\leftrightarrow} v\), \(u\stackrel{\bar{\Delta} }{\nlongleftrightarrow} g \) and \(\calA_2\) implies (using the max-flow min-cut theorem) that there exists a contour \(\gamma\subset E_{\bar{\Delta}_g}\) with:
		\begin{enumerate}[label=(\roman*)]
			\item \(u\xleftrightarrow{\mathring{\gamma}}v,\ g\notin \mathring{\gamma}\),
			\item \(\{e\in\gamma: (\bar{\bfn}+\bar{\bfm})_e>0\}\subset \calC_{\partial}\),
			\item \( |\{e\in\gamma: (\bar{\bfn}+\bar{\bfm})_e>0\}|<\frac{K}{\log(K)^2}\).
		\end{enumerate}
		We will use a many-to-one argument.	Pick an arbitrary order on the set of contours included in \(E_{\bar{\Delta}_g}\) and define the outermost good contour \(\Gamma\) of \(\bar{\bfn}+\bar{\bfm}\) as follows:
		\begin{enumerate}
			\item pick the smallest contour satisfying the three previous conditions,
			\item set all edges of this contour to zero, denote \(\bar{\bfn}^*,\bar{\bfm}^* \) the configurations obtained that way,
			\item let \(\Gamma\) be the outermost closed contour surrounding \(u,v\) in \( \bar{\bfn}^*+\bar{\bfm}^*\).
		\end{enumerate}
		Then define the map \(Y: (\bar{\bfn},\bar{\bfm})\mapsto (\bar{\bfn}',\bar{\bfm}') \) that
		\begin{itemize}
			\item turns the edges of \(\Gamma(\bar{\bfn}+\bar{\bfm})\) to \(0\) in both configurations, call the obtained configurations \(\bar{\bfn}',\bar{\bfm}'\).
			\item for each \(i\in \Gamma^{\exterior}\) do \begin{itemize}
				\item if \(I_i(\bar{\bfn}')\) is odd, change \(\bar{\bfn}'_{ig}=1\) and same for \(\bar{\bfm}'\).
			\end{itemize}
		\end{itemize}
		Notice that, by construction, \(\bar{\bfn}'\) and \(\bar{\bfm}'\) only differ from \(\bar{\bfn}\) and \(\bar{\bfm}\) on \(\Gamma\) and on \(\big\{\{i,g\}\big\}_{i\in\Gamma^{\exterior}}\) and that \( \Gamma\) is the outermost closed contour of \(\bar{\bfn}'+\bar{\bfm}' \). Moreover, all sources of \(\bar{\bfn}'\) and \(\bar{\bfm}'\) are in \(\Gamma^{\interior}\) and the number of edges modified is smaller or equal to \(2\frac{K}{\log(K)^2}\).
		We then use a one-to-many argument: given \((\bar{\bfn},\bar{\bfm})\) admissible,
		\begin{equation*}
		\weight(\bar{\bfn})\weight(\bar{\bfm})\leq e^{cK/\log(K)^2}\weight(\bar{\bfn}')\weight(\bar{\bfm}'),
		\end{equation*}
		as all changes are local and thus have a bounded cost. Then, given \((\bar{\bfn}',\bar{\bfm}')\), one can reconstruct the associated \(\Gamma\) (as it is the outermost closed contour surrounding \(u,v\)). The number of pre-images of \((\bar{\bfn}',\bar{\bfm}')\) is thus bounded by the number of choices for the edges that were open in \((\bar{\bfn},\bar{\bfm})\), which is bounded by
		\begin{align*}
		\binom{cK^d}{\frac{K}{\log(K)^2}}&\leq \big(ecK^{d-1}\log(K)^2\big)^{K/\log(K)^2}\\
		&\leq \big(e^{c\log(K)}\big)^{K/\log(K)^2} = e^{cK/\log(K)}.
		\end{align*}
	\end{case}
	Using those observations,
	\begin{align*}
		&\RCPF_{\Lambda_g}(\varnothing)^2\big\{ \mathds{1}_{u\stackrel{\Delta}{\leftrightarrow} v}(\bar{\bfm})\mathds{1}_{u\stackrel{\bar{\Delta} }{\nlongleftrightarrow} g }(\bar{\bfn}+\bar{\bfm}) \mathds{1}_{\calA_2}(\bar{\bfn}+\bar{\bfm}) \big\} \leq \\
		&\ \leq \sum_{\gamma} \sum_{\substack{\partial\bar{\bfn}=\varnothing\\ \partial\bar{\bfm}=\varnothing}}\sum_{\bar{\bfn}',\bar{\bfm}'} \weight(\bar{\bfn})\weight(\bar{\bfm})\mathds{1}_{u\stackrel{\Delta}{\leftrightarrow} v}(\bar{\bfm}') \mathds{1}_{\Gamma(\bar{\bfn}+\bar{\bfm})=\gamma}\mathds{1}_{Y(\bar{\bfn},\bar{\bfm})=(\bar{\bfn}',\bar{\bfm}')}\\
		&\ \leq \sum_{\gamma} \sum_{\bfn',\bfm'}e^{c K/\log(K)^2}\weight(\bar{\bfn}')\weight(\bar{\bfm}')\mathds{1}_{u\stackrel{\mathring{\gamma}}{\leftrightarrow} v}(\bar{\bfm}') \mathds{1}_{\Gamma(\bar{\bfn}'+\bar{\bfm}')=\gamma}e^{c'K/\log(K)}\\
		&\ \leq e^{cK/\log(K)}\sum_{\gamma}\RCPF_{\Lambda_g\setminus\mathring{\gamma}}(\varnothing)^2\big\{\mathds{1}_{\Gamma=\gamma} \big\} \times \\ &\qquad
		\times \sum_{\substack{A,B\subset\gamma^{\interior}\\|A|,|B|\leq K/\log(K)^2}} \RCPF_{\mathring{\gamma}}(A)\RCPF_{\mathring{\gamma}}(B)\big\{\mathds{1}_{u\leftrightarrow v}(\bfm')\big\},
	\end{align*}
	where in the second line we used that configurations are not modified inside \(\Gamma\) and the last line is the partitioning over the sources of \(\bar{\bfn}'\) and \(\bar{\bfm}'\).
	
	Then, using the Switching Lemma~\ref{lem:RCswitchingLemma},
	\begin{align*}
	\RCPF_{\mathring{\gamma}}(A)\RCPF_{\mathring{\gamma}}(B)\big\{\mathds{1}_{u\leftrightarrow v}(\bfm)\big\} &= \langle\sigma_A\rangle_{\mathring{\gamma}}\RCPF_{\mathring{\gamma}}(\varnothing)\sum_{\partial\bfm=B} \weight(\bfm)\mathds{1}_{u\leftrightarrow v}(\bfm)\\ &\leq\RCPF_{\mathring{\gamma}}(\varnothing)\RCPF_{\mathring{\gamma}}(B)\big\{\mathds{1}_{u\leftrightarrow v}\big\}\\
	&=\RCPF_{\mathring{\gamma}}(u,v)\RCPF_{\mathring{\gamma}}\big(B\Delta\{u,v\}\big)\\
	&\leq \RCPF_{\mathring{\gamma}}(u,v)\RCPF_{\mathring{\gamma}}(\varnothing).
	\end{align*}
	Plugging that in the previous computation yields:
	\begin{align*}
		... \leq &e^{cK/\log(K)}\sum_{\gamma}\RCPF_{\Lambda_g\setminus\mathring{\gamma}}(\varnothing)^2\big\{\mathds{1}_{\Gamma=\gamma} \big\} \RCPF_{\mathring{\gamma}}(u,v)\RCPF_{\mathring{\gamma}}(\varnothing) \Big(\sum_{k=0}^{K/\log(K)^2}\binom{|\partial^{\interior}\gamma|}{k}\Big)^2\\
		&\leq e^{cK/\log(K)}\RCPF_{\Lambda_g}(\varnothing)\RCPF_{\Lambda_g}(u,v)\big\{\mathds{1}_{u\nleftrightarrow g,\partial\bar{\Delta}} \big\}\\
		&\leq e^{cK/\log(K)}\RCPF_{\Lambda_g}(\varnothing)\RCPF_{\Lambda_g}(u,v)\big\{\mathds{1}_{u\nleftrightarrow g} \big\}.
	\end{align*}
	Dividing everything by \(\RCPF_{\Lambda_g}(\varnothing)^2 \) one gets:
	\begin{align*}
		\RCLaw_{\Lambda_g}^{\varnothing,\varnothing}\big( u\stackrel{\Delta}{\leftrightarrow} v \text{ in } \bfm, u\stackrel{\bar{\Delta} }{\nlongleftrightarrow} g \text{ in } \bfn+\bfm, \calA_2 \big)&\leq e^{cK/\log(K)}\langle\sigma_u;\sigma_v\rangle_{\Lambda,h}\\
		&\leq e^{-K(1-c/\log(K))}
	\end{align*}
	once \(\Lambda\) is large enough, since \(v\in \partial^{\exterior}\CGballK(u)\).
\end{proof}

\section{Acknowledgments}
The author thanks Franco Severo for showing him the argument used in the proof of Theorem~\ref{thm:expMixRC}, Yvan Velenik for various comments and corrections on previous drafts of this paper and the anonymous referees for comments that helped improving the overall presentation. The author gratefully acknowledge the support of the Swiss National Science Foundation through the NCCR SwissMAP.

\appendix

\section{A Few Random Current Properties}
We collect here a few properties of the random current together with proofs.

\subsection{Insertion Tolerance}

\begin{lemma}
	\label{lem:RCInsTol}
	For any graph \(G=(V_G,\bfJ)\) and any \(e\in E_G\), uniformly over the values of \(n_f, f\neq e\), one has:
	\begin{equation*}
	\RCLaw^A_{G}(\bfn_e>0 ,\bfn_f=n_f\forall f\neq e) \geq c(\tilde{J}_e)\RCLaw^A_{G}(\bfn_f=n_f\forall f\neq e)
	\end{equation*}
	where \(c(\tilde{J}_e)=\frac{\cosh(\tilde{J}_e)-1}{\cosh(\tilde{J}_e)}\).
\end{lemma}
\begin{proof}
	If the values \(n_f\) implies \(\bfn_e=1\mod{2}\), then \(\RCLaw^A_{G}(\bfn_e>0 \given\bfn_f=n_f\forall f\neq e)=1\) and it is over. Otherwise, \(\RCLaw^A_{G}(\bfn_e>0 \given\bfn_f=n_f\forall f\neq e)\) is the probability for a Poisson random variable of parameter \(\tilde{J}_e\) to be positive conditionally on being even.
\end{proof}

\subsection{Exponential Ratio Mixing when \(h>0\) }

We describe here an adaptation of an argument due to Duminil-Copin~\cite{Duminil-Copin-2018} to obtain exponential mixing under the random current measure with a field (a version of this idea is used in \cite{duminil-copin_exponential_2018}).

We describe the results for a finite weighted graph \(\Lambda\). As we consider random current measures, the support of a local event is a set of edges; to handle distances between supports define, for \(E\subset E_{\Lambda_g}\),
\begin{equation}
	V=V(E) = \bigcup_{e\in E}e\cap V_{\Lambda}
\end{equation}
the set of non-ghost endpoints of edges in \(E\).

\begin{theorem}
	\label{thm:expMixRC}
	There exist \(R\geq 0\) and \(C\geq 0\) such that, for any \(E_1,E_2\) sets of edges, \(V_i=V(E_i) \), any \(A\subset V_{\Lambda}\) and any events \(D\) \(E_1\)-measurable and \(D'\) \(E_2\)-measurable, if \(d_{\Lambda}(A\cup V_1,V_2) >R\), then
	\begin{equation}
	\Big|\log\Big( \frac{\RCLaw_{\Lambda_g}^{A} (D,D')}{\RCLaw_{\Lambda_g}^{A} (D)\RCLaw_{\Lambda_g}^{A} (D')}\Big)\Big|\leq C\cosh(h)^{-d_{\Lambda}(A\cup V_1,V_2)}.
	\end{equation}
\end{theorem}
\begin{proof}
	Fix two disjoint sets of edges \(E_1\) and \(E_2\). Denote \(\tilde{\Lambda}_g\) the graph obtained from \(\Lambda_g\) by removing the edges of \(E_1\cup E_2\). Then the key observation is that for any configurations \( n_1, m_1\in \N^{E_1}\) and \( n_2, m_2\in \N^{E_2}\),
	\begin{equation*}
	\frac{\RCLaw_{\Lambda_g}^{A} ( n_1, n_2)\RCLaw_{\Lambda_g}^{A} ( m_1, m_2)}{\RCLaw_{\Lambda_g}^{A} ( n_1, m_2)\RCLaw_{\Lambda_g}^{A} ( m_1, n_2)} = \frac{\RCPF_{\tilde{\Lambda}_g}(A\Delta\partial n_1\Delta\partial n_2)\RCPF_{\tilde{\Lambda}_g}(A\Delta\partial m_1\Delta\partial m_2)}{\RCPF_{\tilde{\Lambda}_g}(A\Delta\partial n_1\Delta\partial m_2)\RCPF_{\tilde{\Lambda}_g}(A\Delta\partial m_1\Delta\partial n_2)}.
	\end{equation*}
	Then, the RHS can be written
	\begin{multline*}
	\frac{\RCPF_{\tilde{\Lambda}_g}(A\Delta\partial n_1\Delta\partial n_2)\RCPF_{\tilde{\Lambda}_g}(\varnothing)}{\RCPF_{\tilde{\Lambda}_g}(A\Delta\partial n_1)\RCPF_{\tilde{\Lambda}_g}(\partial n_2)} \frac{\RCPF_{\tilde{\Lambda}_g}(A\Delta\partial m_1\Delta\partial m_2)\RCPF_{\tilde{\Lambda}_g}(\varnothing)}{\RCPF_{\tilde{\Lambda}_g}(A\Delta\partial m_1)\RCPF_{\tilde{\Lambda}_g}(\partial m_2)} \times\\
	\times \frac{\RCPF_{\tilde{\Lambda}_g}(A\Delta\partial n_1)\RCPF_{\tilde{\Lambda}_g}(\partial m_2)}{\RCPF_{\tilde{\Lambda}_g}(A\Delta\partial n_1\Delta\partial m_2)\RCPF_{\tilde{\Lambda}_g}(\varnothing)} \frac{\RCPF_{\tilde{\Lambda}_g}(A\Delta\partial m_1)\RCPF_{\tilde{\Lambda}_g}(\partial n_2)}{\RCPF_{\tilde{\Lambda}_g}(A\Delta\partial m_1\Delta\partial n_2)\RCPF_{\tilde{\Lambda}_g}(\varnothing)}.
	\end{multline*}
	Now, using the switching lemma,
	\begin{equation*}
	1-\frac{\RCPF_{\tilde{\Lambda}_g}(A\Delta\partial n_1)\RCPF_{\tilde{\Lambda}_g}(\partial n_2)}{\RCPF_{\tilde{\Lambda}_g}(A\Delta\partial n_1\Delta\partial n_2)\RCPF_{\tilde{\Lambda}_g}(\varnothing)}= \RCLaw_{\tilde{\Lambda}_g}^{A\Delta\partial n_1\Delta\partial n_2,\varnothing}\big( \evenPart_{\partial n_2}^c \big).
	\end{equation*}
	Now, if \((A\Delta\partial n_1)\cap\partial n_2\cap\Lambda\neq\varnothing\), then simply bound the probability by \(1\). Otherwise, \((A\Delta\partial n_1)\cap\partial n_2=\{g\}\) or \(\varnothing\). In both cases, the combination of the sources constraint and \( \evenPart_{\partial n_2}^c\) implies the existence of an edge-self-avoiding path \(\gamma\) going from \((A\Delta\partial n_1)\) to \(\partial n_2\) in the first current and such that \(\gamma\nleftrightarrow g\) in the sum of the two currents. One thus gets,
	\begin{align*}
	\RCLaw_{\tilde{\Lambda}_g}^{A\Delta\partial n_1\Delta\partial n_2,\varnothing}\big( \evenPart_{\partial n_2}^c \big)&\leq \cosh(h)^{-d_{\Lambda}\big((A\Delta\partial n_1)\cap\Lambda,\partial n_2\cap\Lambda\big)}\\
	&\leq \cosh(h)^{-d_{\Lambda}(A\cup V_1,V_2)},
	\end{align*}as \(d_{\Lambda}\big((A\Delta\partial n_1)\cap\Lambda,\partial n_2\cap\Lambda\big)\geq d_{\Lambda}(A\cup V_1,V_2) \), where \(V_i=\bigcup_{e\in E_1}e\cap\Lambda \). Using this, there exist \(C\geq 0\) and \(R\geq 0\) such that
	\begin{equation*}
	\Big|\log\Big(\frac{\RCLaw_{\Lambda_g}^{A} ( n_1, n_2)\RCLaw_{\Lambda_g}^{A} ( m_1, m_2)}{\RCLaw_{\Lambda_g}^{A} ( n_1, m_2)\RCLaw_{\Lambda_g}^{A} ( m_1, n_2)}\Big)\Big|\leq C\cosh(h)^{-d_{\Lambda}(A\cup V_1,V_2)}
	\end{equation*}
	whenever \(d_{\Lambda}(A\cup V_1,V_2)\geq R\). Now, fix \(A\subset\Lambda\), take \(E_1,E_2\) two sets of edges, let \(V_1,V_2\) be defined as before. Suppose \(d_{\Lambda}(A\cup V_1,V_2)= L>R\), then for any two events \(D,D'\) supported on \(E_1,E_2\) respectively,
	\begin{align*}
	\RCLaw_{\Lambda_g}^{A} (D,D') &= \sum_{\substack{ n_1\in \N^{E_1}\\  n_1\in D}}\sum_{\substack{ n_2\in \N^{E_2}\\  n_2\in D'}}\RCLaw_{\Lambda_g}^{A} ( n_1, n_2)\sum_{ m_1\in \N^{E_1}}\sum_{ m_2\in \N^{E_2}}\RCLaw_{\Lambda_g}^{A} ( m_1, m_2)\\
	&\leq \exp{C\cosh(h)^{-L}}\sum_{\substack{ n_1\in D\\ n_2\in D' }}\sum_{ m_1,  m_2 }\RCLaw_{\Lambda_g}^{A} ( n_1, m_2)\RCLaw_{\Lambda_g}^{A} ( m_1, n_2)\\
	&= \exp{C\cosh(h)^{-L}}\RCLaw_{\Lambda_g}^{A} (D)\RCLaw_{\Lambda_g}^{A} (D').
	\end{align*}In the same fashion,
	\begin{equation*}
	\RCLaw_{\Lambda_g}^{A} (D,D') \geq \exp{-C\cosh(h)^{-L}}\RCLaw_{\Lambda_g}^{A} (D)\RCLaw_{\Lambda_g}^{A} (D').
	\end{equation*}
	So,
	\begin{equation*}
	\Big|\log\Big( \frac{\RCLaw_{\Lambda_g}^{A} (D,D')}{\RCLaw_{\Lambda_g}^{A} (D)\RCLaw_{\Lambda_g}^{A} (D')}\Big)\Big|\leq C\cosh(h)^{-L}.
	\end{equation*}
\end{proof}

Using the same technique, one can obtain:
\begin{lemma}
	\label{lem:sourceToSourceless}
	There exist \(R\geq 0\) and \(C\geq 0\) such that, for any \(E\) set of edges, \(V=V(E) \), any \(A\subset V_{\Lambda}\) and any event \(D\) \(E\)-measurable, if \(d_{\Lambda}(A,V) >R\), then
	\begin{equation}
	\Big|\log\Big( \frac{\RCLaw_{\Lambda_g}^{A} (D)}{\RCLaw_{\Lambda_g}^{\varnothing} (D)}\Big)\Big|\leq C\cosh(h)^{-d_{\Lambda}(A,V)}.
	\end{equation}
\end{lemma}
\begin{proof}
	As before, let \( n\in\N^{E}\), and let \(\tilde{\Lambda}_g\) be the graph obtained by removing edges in \(E\) from \(\Lambda_g\). Then
	\begin{equation*}
		\RCLaw_{\Lambda_g}^{A}( n) = \sum_{ m\in\N^{E} } \frac{\RCLaw_{\Lambda_g}^{A}( n)\RCLaw_{\Lambda_g}^{\varnothing}( m)}{\RCLaw_{\Lambda_g}^{\varnothing}( n)\RCLaw_{\Lambda_g}^{A}( m)}\RCLaw_{\Lambda_g}^{\varnothing}( n)\RCLaw_{\Lambda_g}^{A}( m).
	\end{equation*}
	So, one just need to control the fraction term:
	\begin{equation*}
		\frac{\RCLaw_{\Lambda_g}^{A}( n)\RCLaw_{\Lambda_g}^{\varnothing}( m)}{\RCLaw_{\Lambda_g}^{A}( m)\RCLaw_{\Lambda_g}^{\varnothing}( n)} = \frac{\RCPF_{\tilde{\Lambda}_g}(A\Delta\partial n)\RCPF_{\tilde{\Lambda}_g}(\partial m)}{\RCPF_{\tilde{\Lambda}_g}(A\Delta\partial m)\RCPF_{\tilde{\Lambda}_g}(\partial n)}.
	\end{equation*}
	Proceeding as in the proof of Theorem~\ref{thm:expMixRC}, one get the wanted estimate.
\end{proof}

\section{Toolbox}

\subsection{A Combinatorial Lemma}

\begin{lemma}
	\label{lem:sourceGraph}
	Let \(G=(V_G,E_G)\) be a finite connected graph. For any \(A\subset V_G\) of even cardinality, there exists \(\omega\subset E_G\) with \(\partial\omega = A\).
\end{lemma}
\begin{proof}
	As \(G\) is connected, it admits a spanning tree. So it is sufficient to prove the result for trees. Assume \(G\) is a tree. Let \(A=\{a_1,... ,a_n\}\). We proceed by induction over \(n=|A|\). For \(n=2\), set \(\omega\) to be the (unique) path going from \(a_1\) to \(a_2\). For \(n\) even, suppose one has constructed \(\omega'\) with \(\partial\omega' = \{a_1,... ,a_{n-2}\}\). Let \(\gamma\) be the unique path going from \(a_{n-1}\) to \(a_n\) in \(G\). Set \(\omega = \omega'\Delta \gamma\). As the sources of the symmetric difference is the symmetric difference of the sources, we have \(\partial\omega= \partial\omega'\Delta\partial\gamma=\{a_1,... ,a_{n-2}\}\Delta\{a_{n-1},a_n\}= A\).
\end{proof}

\subsection{A Geometrical Lemma}

For \(\icl\) a norm on \(\R^d\), \(t\in\partial\WulffShape\) (see Subsection~\ref{subsec:DiamondDecomp_and_MainThm}) and \(\delta\in(0,1)\), define cones
\begin{equation*}
\fcone_{\delta}=\big\{ x\in\R^d: (t,x)_d>(1-\delta)\icl(x) \big\},\qquad \bcone_{\delta} = -\fcone_{\delta}(t).
\end{equation*}

Remark that \(\fcone_{\delta}\) are increasing sets in \(\delta\). Let \(A\) be a compact subset of \(\R^d\) and \(x\in \R^d\). We say that \(A\) \(\delta\)\emph{-sees} \(x\) if there exists \(y\in A\) with \(x\in y+\fcone_{\delta}\); we say that \(A\) \(\delta\)\emph{-blocks} \(x\) if \(A\not\subset x+\bcone_{\delta}\) (in other words, \(A\) \(\delta\)-blocks \(x\) if \(x\) does not \(\delta\)-backward-see \(A\)).

\begin{lemma}
	\label{lem:createdCPvolume}
	Let \(A\) be a bounded subset of \(\R^d\). Let \(\delta,\delta'>0\) such that \(\delta+\delta'<1\). Then, the diameter of \[V=\{x\in\R^d: A\ \delta\text{-sees}\ x, A\ (\delta+\delta')\text{-blocks}\ x \}\] is upper bounded by a constant depending only on \(\delta,\delta',A,d\).
\end{lemma}
\begin{proof}
	As the only parameters of our problem are \(\delta,\delta',A,d\), one only need to show that \(V\) is bounded. The first observation is that if \(x\) \(\delta\)-sees \(A\), then \(x\) \(\delta\)-sees \(y\) for any \(y\) \(\delta\)-seen by \(A\); indeed, if \(y\in x+\fcone_{\delta}\) then \((y+\fcone_{\delta} )\subset (x+\fcone_{\delta})\). The second observation is that if \(x\) is not \(\delta\)-blocked by \(A\), then so are all \(y\in x+\fcone_{\delta}\). Now, as \(A\) is bounded, there exist \(a,b\in\R^d\) such that
	\begin{itemize}
		\item \(a\) \(\delta\)-sees \(A\),
		\item \(b\) is not \((\delta+\delta')\)-blocked by \(A\).
	\end{itemize}
	The two observations made before imply that \(V\) is a subset of \((a+\fcone_{\delta}) \cap (b+\fcone_{\delta+\delta'})^c\). The final observation is that, as \(\delta'>0\), the previous set is bounded. This implies the Lemma.
\end{proof}

\bibliographystyle{plain}
\bibliography{OZIsingPosh}

\end{document}